\documentclass[acmsmall,table,xcdraw,screen]{acmart}

\setcopyright{rightsretained}
\acmPrice{}
\acmDOI{10.1145/3571234}
\acmYear{2023}
\copyrightyear{2023}
\acmSubmissionID{popl23main-p278-p}
\acmJournal{PACMPL}
\acmVolume{7}
\acmNumber{POPL}
\acmArticle{41}
\acmMonth{1}
\received{2022-07-07}
\received[accepted]{2022-11-07}

\usepackage[capitalise]{cleveref}

\usepackage{subcaption} %% For complex figures with subfigures/subcaptions
\usepackage{booktabs}   %% For formal tables:
\usepackage{algorithm}
\usepackage{algpseudocode}

\usepackage{diagbox} % TODO is adding OK
\usepackage{arydshln} % TODO as above

\usepackage{mathpartir}
\usepackage[clock]{ifsym}  % for stopwatch in ablation table
\usepackage{enumitem}
\usepackage{tikz}
\usepackage{multirow}
\usepackage{makecell}
\usepackage{float}
\floatstyle{plaintop}
\restylefloat{table}

\newcommand{\ram}{{
\begin{tikzpicture}[thick,scale=0.6, every node/.style={scale=0.6}, baseline={([yshift=-.8ex]current bounding box.center)}]
\protect\node[draw] at (0, 0) {RAM};
\protect\draw (-0.35,-0.35) -- (0.35,0.35);
\end{tikzpicture}
}}

\newcommand{\hole}{\texttt{??}}
\newcommand{\stitch}{\textsc{Stitch}}
\newcommand{\abshole}{A_{\hole}} % :)

\newcommand{\cost}[1]{\text{cost}(#1)}
\newcommand{\coststar}[1]{\text{cost}_{\alpha=0}(#1)}
\newcommand{\costt}[1]{\text{cost}_t(#1)}
\newcommand{\costlam}{\text{cost}_\lambda}
\newcommand{\costapp}{\text{cost}_\texttt{app}}
\newcommand{\costvar}{\text{cost}_\texttt{\$i}}
\newcommand{\costabs}{\text{cost}_\alpha}

\newcommand{\subst}{ \circ }

\newcommand{\gsym}{\mathcal{G}_\texttt{sym}}

\newcommand{\corpus}{\mathcal{P}}
\newcommand{\rewritestrategy}{\mathcal{R}}
\newcommand{\util}[1]{U_{\corpus,\rewritestrategy}(#1)}

\newcommand{\upperbound}[1]{\bar{U}_{\corpus,\rewritestrategy}(#1)}

\newcommand{\rewrite}[1]{\textsc{Rewrite}_\rewritestrategy(\corpus,#1)}
\newcommand{\rewriteone}[2]{\textsc{RewriteOne}(#1,#2)} %\rewriteone{e}{A}

\newcommand{\rewritelocations}[1]{\textsc{RewriteLocations}_\rewritestrategy(\corpus,#1)}

\newcommand{\matches}[1]{\textsc{Matches}(\corpus,#1)}

\newcommand{\blue}[1]{\textcolor{blue}{#1}}

\newcommand{\matt}[1]{}

\newcommand{\armandotodo}[1]{}

\newcommand{\para}[1]{\paragraph{\textbf{#1}}}

\newcommand{\eval}[1]{\llbracket{#1}\rrbracket}

\newtheorem{thm}{Theorem}
\newtheorem{lemma}[thm]{Lemma}
\definecolor{StitchDarkBlue}{RGB}{19, 39, 171}
\definecolor{StitchLightBlue}{RGB}{24, 156, 196}
\definecolor{StitchPurple}{RGB}{141, 46, 166}

\definecolor{StitchOrange}{RGB}{255, 171, 64}
\definecolor{StitchGreen}{RGB}{106, 168, 79}

\newcommand{\clk}{{ \raisebox{-.4ex}{ \includegraphics[height=12pt]{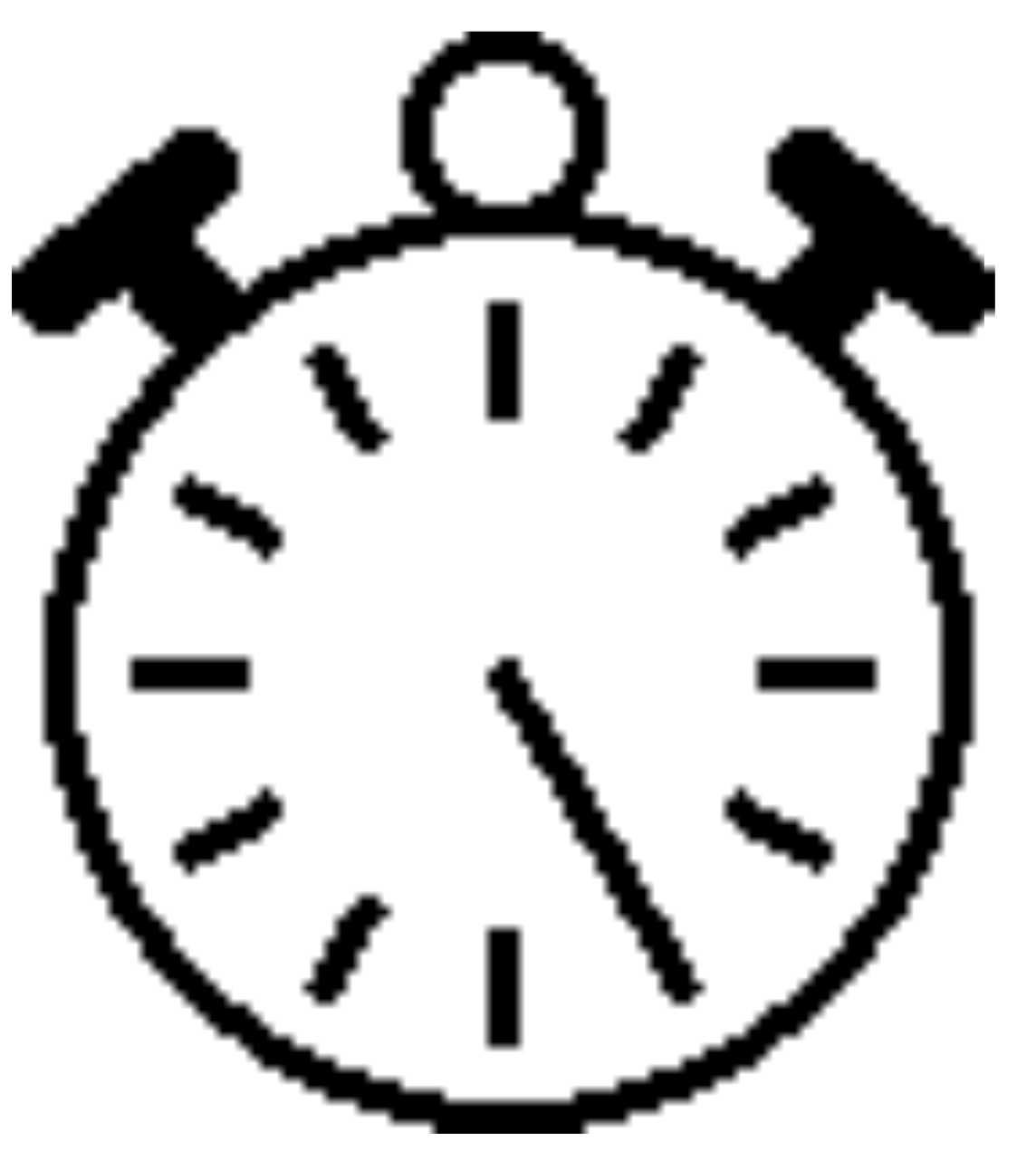}} }}
\newcommand{\clksmol}{\clk}

\begin{document}

%% Title information
\title{Top-Down Synthesis for Library Learning}         %% [Short Title] is optional;
                                        %% when present, will be used in
                                        %% header instead of Full Title.
% \titlenote{with title note}             %% \titlenote is optional;
                                        %% can be repeated if necessary;
                                        %% contents suppressed with 'anonymous'
% \subtitle{Subtitle}                     %% \subtitle is optional
% \subtitlenote{with subtitle note}       %% \subtitlenote is optional;
                                        %% can be repeated if necessary;
                                        %% contents suppressed with 'anonymous'

%% Author information
%% Contents and number of authors suppressed with 'anonymous'.
%% Each author should be introduced by \author, followed by
%% \authornote (optional), \orcid (optional), \affiliation, and
%% \email.
%% An author may have multiple affiliations and/or emails; repeat the
%% appropriate command.
%% Many elements are not rendered, but should be provided for metadata
%% extraction tools.

%% Author with single affiliation.
\author{Matthew Bowers}
% \authornote{with author1 note }          %% \authornote is optional;
                                        %% can be repeated if necessary
% \orcid{0000-0001-8450-7033}             %% \orcid is optional
\affiliation{
%   \position{Position1}
%   \department{EECS}              %% \department is recommended
  \institution{Massachusetts Institute of Technology}            %% \institution is required
%   \streetaddress{Street1 Address1}
  \city{Cambridge}
  \state{MA}
%   \postcode{Post-Code1}
  \country{USA}                    %% \country is recommended
}
\email{mlbowers@mit.edu}          %% \email is recommended

%% Author with two affiliations and emails.
\author{Theo X. Olausson}

\affiliation{
%   \position{Position1}
%   \department{EECS}              %% \department is recommended
  \institution{Massachusetts Institute of Technology}            %% \institution is required
%   \streetaddress{Street1 Address1}
  \city{Cambridge}
  \state{MA}
%   \postcode{Post-Code1}
  \country{USA}                    %% \country is recommended
}
\email{theoxo@mit.edu}          %% \email is recommended

\author{Lionel Wong}

\affiliation{
%   \position{Position1}
%   \department{EECS}              %% \department is recommended
  \institution{Massachusetts Institute of Technology}            %% \institution is required
%   \streetaddress{Street1 Address1}
  \city{Cambridge}
  \state{MA}
%   \postcode{Post-Code1}
  \country{USA}                    %% \country is recommended
}
\email{zyzzyva@mit.edu}          %% \email is recommended

\author{Gabriel Grand}

\affiliation{
%   \position{Position1}
%   \department{EECS}              %% \department is recommended
  \institution{Massachusetts Institute of Technology}            %% \institution is required
%   \streetaddress{Street1 Address1}
  \city{Cambridge}
  \state{MA}
%   \postcode{Post-Code1}
  \country{USA}                    %% \country is recommended
}

\email{grandg@mit.edu}          %% \email is recommended

\author{Joshua B. Tenenbaum}

\affiliation{
%   \position{Position1}
%   \department{EECS}              %% \department is recommended
  \institution{Massachusetts Institute of Technology}            %% \institution is required
%   \streetaddress{Street1 Address1}
  \city{Cambridge}
  \state{MA}
%   \postcode{Post-Code1}
  \country{USA}                    %% \country is recommended
}

\author{Kevin Ellis}

\affiliation{
%   \position{Position1}
%   \department{EECS}              %% \department is recommended
  \institution{Cornell University}            %% \institution is required
%   \streetaddress{Street1 Address1}
  \city{Ithaca}
  \state{NY}
%   \postcode{Post-Code1}
  \country{USA}                    %% \country is recommended
}
% \email{kellis@cornell.edu}          %% \email is recommended

\author{Armando Solar-Lezama}

\affiliation{
%   \position{Position1}
%   \department{EECS}              %% \department is recommended
  \institution{Massachusetts Institute of Technology}            %% \institution is required
%   \streetaddress{Street1 Address1}
  \city{Cambridge}
  \state{MA}
%   \postcode{Post-Code1}
  \country{USA}                    %% \country is recommended
}

                                        %% can be repeated if necessary
% \orcid{nnnn-nnnn-nnnn-nnnn}             %% \orcid is optional
% \affiliation{
%   \position{Position2a}
%   \department{Department2a}             %% \department is recommended
%   \institution{Institution2a}           %% \institution is required
%   \streetaddress{Street2a Address2a}
%   \city{City2a}
%   \state{State2a}
%   \postcode{Post-Code2a}
%   \country{Country2a}                   %% \country is recommended
% }
% \email{first2.last2@inst2a.com}         %% \email is recommended
% \affiliation{
%   \position{Position2b}
%   \department{Department2b}             %% \department is recommended
%   \institution{Institution2b}           %% \institution is required
%   \streetaddress{Street3b Address2b}
%   \city{City2b}
%   \state{State2b}
%   \postcode{Post-Code2b}
%   \country{Country2b}                   %% \country is recommended
% }
% \email{first2.last2@inst2b.org}         %% \email is recommended

\renewcommand{\shortauthors}{M. Bowers, T. X. Olausson, L. Wong, G. Grand, J. B. Tenenbaum, K. Ellis, A. Solar-Lezama}

%% Abstract
%% Note: \begin{abstract}...\end{abstract} environment must come
%% before \maketitle command
\begin{abstract}

This paper introduces \textit{corpus-guided top-down synthesis} as a mechanism for synthesizing library functions that capture common functionality from a corpus of programs in a domain specific language (DSL). The algorithm builds abstractions directly from initial DSL primitives, using syntactic pattern matching of intermediate abstractions to intelligently prune the search space and guide the algorithm towards abstractions that maximally capture shared structures in the corpus. We present an implementation of the approach in a tool called \textsc{Stitch} and evaluate it against the state-of-the-art deductive library learning algorithm from DreamCoder. Our evaluation shows that \textsc{Stitch} is 3-4 orders of magnitude faster and uses 2 orders of magnitude less memory while maintaining comparable or better library quality (as measured by compressivity). We also demonstrate \textsc{Stitch}'s scalability on corpora containing hundreds of complex programs that are intractable with prior deductive approaches and show empirically that it is robust to terminating the search procedure early---further allowing it to scale to challenging datasets by means of early stopping.

\end{abstract}

%% 2012 ACM Computing Classification System (CSS) concepts
%% Generate at 'http://dl.acm.org/ccs/ccs.cfm'.

%\ccsdesc[500]{Theory of computation~Branch-and-bound}
%% End of generated code

%% Keywords
%% comma separated list
  %% \keywords are mandatory in final camera-ready submission

%% \maketitle1
%% Note: \maketitle command must come after title commands, author
%% commands, abstract environment, Computing Classification System
%% environment and commands, and keywords command.
\maketitle

% \pagebreak
\section{Introduction}\label{sec:introduction}

\begin{figure}[t]
    \centering
    \includegraphics[width=\textwidth]{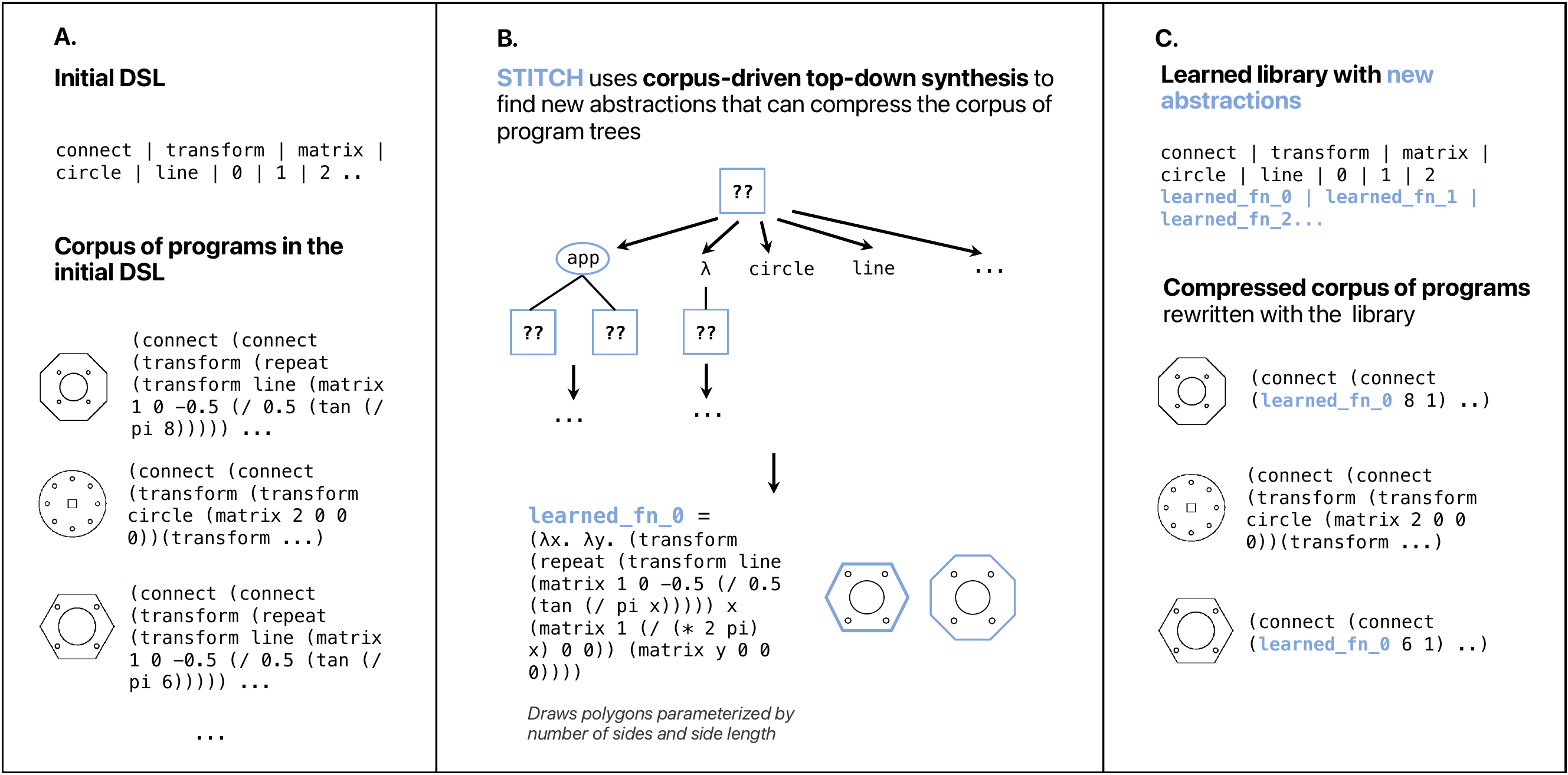}
    \caption{(\textbf{A}) Given an initial DSL of lower-level primitives and a corpus of programs written in the initial DSL, (\textbf{B}) \stitch\ uses a fast and memory efficient \textit{corpus-guided top-down search} algorithm to construct function abstractions which maximally capture shared structure across the corpus of programs. (\textbf{C}) \stitch\ automatically rewrites the initial corpus  using the learned library functions.}
\label{fig:stitch_splash}
% \vspace{-.75cm}
\end{figure}

One way programmers manage complexity is by hiding functionality behind functional abstractions. For example, consider the graphics programs at the bottom of \cref{fig:stitch_splash}A. Each uses a generic set of drawing primitives and renders a technical schematic of a hardware component, shown to the left of each program. Faced with the task of writing \textit{more} of these rendering programs, an experienced human programmer likely would not continue using these low-level primitives. Instead, they would introduce new functional abstractions, like the one at the bottom of \cref{fig:stitch_splash}B which renders a regular polygon given a size and number of sides. Useful abstractions like these allow more concise and legible programs to render the existing schematics. More importantly, well-written abstractions should generalize, making it easier to write new graphics programs for similar graphics tasks.

Recently, the  \textit{program synthesis} community has introduced new approaches that can mimic this process, automatically building a library of functional abstractions in order to tackle more complex synthesis problems \cite{ellis2021dreamcoder, ellis2020dreamcoder,shin2019program,lazaro2019beyond,dechter2013bootstrap}. One popular approach to library learning is to search for common tree fragments across a corpus of programs, which can be introduced as new abstractions \citep{shin2019program,lazaro2019beyond,dechter2013bootstrap}. \citealt{ellis2021dreamcoder}, however, introduces an algorithm that reasons about variable bindings to abstract out well-formed \textit{functions} instead of just tree fragments. While it produces impressive results, the system in \citealt{ellis2021dreamcoder} takes a \textit{deductive} approach to library learning that is difficult to scale to larger datasets of longer and deeper input programs. This approach is deductive in that it uses semantics-preserving rewrite rules to attempt to refactor existing programs to expose shared structure. This requires representing and evaluating an exponentially large space of proposed refactorings to identify common functionality across the corpus. Prior work, such as \citealt{ellis2021dreamcoder, ellis2020dreamcoder}, approaches this challenge by combining a dynamic bottom-up approach to refactoring with version spaces to more efficiently search over the refactored programs. However, these deductive approaches face daunting memory and search requirements as the corpus scales in size and complexity.

This paper introduces an alternate approach to library learning, while preserving the focus on well-formed function abstractions from \citealt{ellis2021dreamcoder}. Instead of taking a \textit{deductive} approach based on refactoring the corpus with rewrite rules, we \textit{directly synthesize abstractions}. We call this approach \textit{corpus-guided top-down synthesis}, and it is based on the insight that when applied to the task of synthesizing abstractions, top-down search can be guided precisely towards discovering shared abstractions over a set of existing training programs. At every step of the search, syntactic comparisons between a partially constructed abstraction and the set of training programs can be used to strongly constrain the search space and direct the search towards abstractions that capture the greatest degree of shared syntactic structure.

We implement this approach in \stitch, a corpus-guided top-down library learning tool written in Rust (Fig. \ref{fig:stitch_splash}). \stitch\ is open-source and available both as a Python package and a Rust crate; the code, installation instructions, and tutorials are available at the GitHub\footnote{\url{https://github.com/mlb2251/stitch}}. We evaluate \stitch\ through a series of experiments (\cref{sec:experiments}), and find that: \stitch\ learns libraries of comparable quality to those found by the algorithm of \citealt{ellis2021dreamcoder} on their iterative bootstrapped library learning task, while being 3-4 orders of magnitude faster and using 2 orders of magnitude less memory (\cref{subsec:claim-1}); \stitch\ learns high quality libraries within seconds to single-digit minutes when run on corpora containing a few hundred programs with mean lengths between 76--189 symbols (sourced from \citealt{cogsci}), while even the simplest of these corpora lies beyond the reach of the algorithm of \citealt{ellis2021dreamcoder} (\cref{subsec:claim-2}); and that \stitch\ degrades gracefully when resources are constrained (\cref{subsec:claim-3}). We also perform ablation studies to expose the relative impact of different optimizations in \stitch\ (\cref{subsec:claim-4}). Finally, we show that \stitch\ is complementary to deductive rewrite approaches (\cref{subsec:claim-5}).

In summary, our paper makes the following contributions:
\begin{enumerate}[itemsep=0.5em]
    \item \textbf{Corpus-guided top-down synthesis (CTS)}: A novel, strongly-guided branch-and-bound algorithm for synthesizing function abstractions  (Sec. \ref{sec:technical}).
    
    \item \textbf{CTS for program compression}: An instantiation of the CTS framework for utility functions favoring abstractions that compress a corpus of programs (Sec. \ref{sec:framework-instantiation}).
    \item \textbf{\stitch}: A parallel Rust implementation of CTS for compression that achieves 3-4 orders of magnitude of speed and memory improvements over prior work, and the analysis of its performance, scaling, and optimizations through several experiments (Sec. \ref{sec:experiments}).

\end{enumerate}

\section{Overview}\label{sec:overview}

\begin{figure}[t]
    \centering
    \includegraphics[width=\textwidth]{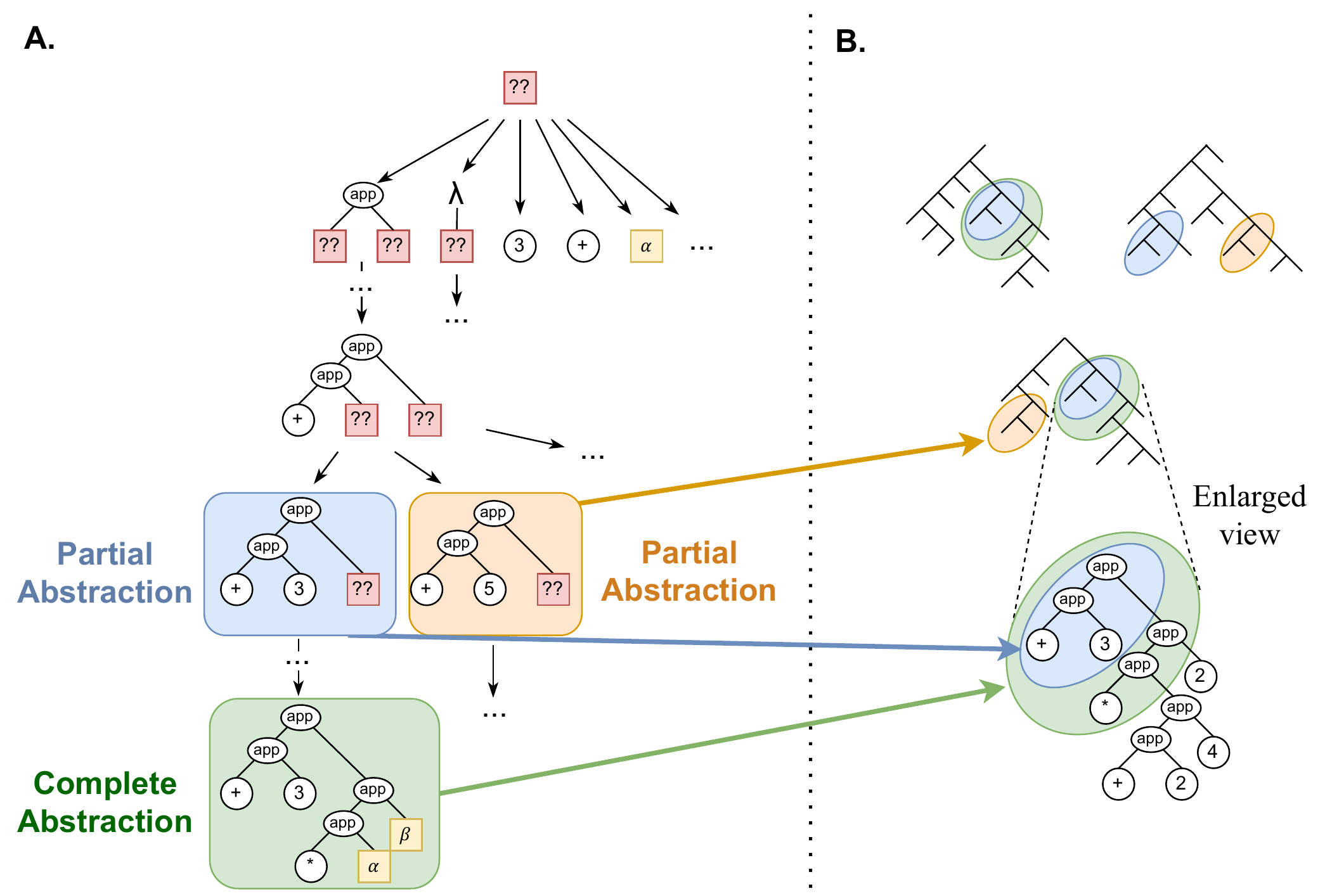}
    \caption{
    (\textbf{A}) Schematic of \textit{corpus-guided top-down search} (CTS). Partial abstractions can contain \textit{holes} indicating unfinished subtrees, denoted \hole, while complete abstractions do not contain holes. (\textbf{B})  Partial and complete abstractions \textit{match} at locations in the corpus. The blue partial abstraction can be expanded into the green complete abstraction so the green abstraction matches in a \textit{subset} of places that the blue abstraction does. For complete abstractions, \textit{matching} indicates that the subtree can be rewritten to use the abstraction.
} \label{fig:corpus_guided}
% \vspace{-.75cm}
\end{figure}

In this section, we build intuition for the algorithmic insights that power \stitch. As a running example, we focus on learning a single abstraction from the following corpus of programs:
\begin{equation}
\begin{aligned}
    &\lambda x.\ \texttt{\blue{+ 3} (\blue{*} (+ 2 4) 2)}\\
    % &\lambda x.\ \texttt{\blue{+ 3} (\blue{*} (+ 2 4) 5)}\\
    &\lambda xs.\ \texttt{map ($\lambda x.\ $\blue{+ 3} (\blue{*} 4 (+ 3 x))) xs}\\
    &\lambda x.\ \texttt{* 2 (\blue{+ 3} (\blue{*} x (+ 2 1)))}\\
\end{aligned}
\label{eq:overview-examples}
\end{equation}

The notion  of what a \textit{good} abstraction is depends on the application, so our algorithm is generic over the utility function that we seek to maximize. Following prior work \citep{ellis2021dreamcoder,shin2019program,dechter2013bootstrap} we focus on \textit{compression} as a utility function: a good abstraction is one which minimizes the size of the corpus when rewritten with the abstraction.
The utility function used by \stitch\ is detailed in \cref{sec:framework-instantiation} and corresponds exactly to the compression objective, but at a high level the function seeks to maximize the product of \textit{the size of the abstraction} and the \textit{number of locations} where the abstraction can be used. This product balances two key features of a highly compressive abstraction: the abstraction should be general enough that it applies in many locations, but specific enough that it captures a lot of structure at each location.

The optimal abstraction that maximizes our utility in this example is:
\begin{equation}
\texttt{\blue{fn0}} = \lambda\alpha.\ \lambda\beta. \texttt{(+ 3 (* $\alpha$ $\beta$))}
\end{equation}
And the shared structure that this abstraction captures is highlighted in blue in \cref{eq:overview-examples}. When rewritten to use this abstraction, the size of the resulting programs is minimized:
\begin{equation}
\begin{aligned}
    &\lambda x.\ \texttt{\blue{fn0} (+ 2 4) 2}\\
    &\lambda xs.\ \texttt{map ($\lambda x.\ $\blue{fn0} 4 (+ 3 x)) xs}\\
    &\lambda x.\ \texttt{* 2 (\blue{fn0} x (+ 2 1))}\\
\end{aligned}
\label{eq:overview-examples-rewritten}
\end{equation}

\stitch\ synthesizes the optimal abstraction directly through \textit{top-down search}. Top-down search methods construct program syntax trees by iteratively refining partially-completed programs that have unfinished \textit{holes}, until a complete program that meets the specification is produced \cite{feser2015synthesizing,polikarpova2016program, balog2016deepcoder,ellis2021dreamcoder,nye2021blended,2020swaratneuraladmissableheuristics}. This kind of search is often made efficient by identifying branches of the search tree that can be pruned away because the algorithm can efficiently determine that none of the  programs in that branch can be correct. We aim to apply a similar idea to the problem of synthesizing good abstractions; the idea is to explore the space of functions in the same top-down way, but in search of a function that maximizes the utility measure. The key observation is that this new objective affords even more aggressive pruning opportunities than the traditional correctness objective, allowing us to synthesize optimal abstractions very efficiently. We call this approach \textit{corpus-guided top-down search} (CTS).

\para{Corpus-Guided Top-Down Search.}
Like other top-down synthesis approaches, our algorithm explores the space of abstractions by repeatedly refining abstractions with holes, as illustrated in Fig. \ref{fig:corpus_guided}A. We call abstractions with holes \textit{partial abstractions} in contrast to \textit{complete abstractions} which have no holes. Our top-down algorithm searches over abstraction \textit{bodies}, so to synthesize the optimal abstraction \texttt{fn0} in our running example, we synthesize its body: \texttt{(+ 3 (* $\alpha$ $\beta$))}.

We say that a partial abstraction \textit{matches} at a subtree in the corpus if there’s a possible assignment to the holes and arguments that yields the subtree. For example, consider the subtree \texttt{(+ 3 (* (+ 2 4) 2)} from the first program in our running example, shown as a syntax tree at the bottom of figure \ref{fig:corpus_guided}B. The partial abstraction \texttt{(+ 3 \hole)} shown in blue matches here with the hole $\hole = \texttt{(* (+ 2 4) 2)}$, and the complete abstraction \texttt{(+ 3 (* $\alpha$ $\beta$))} matches here with $\alpha = \texttt{(+ 2 4)}$ and $\beta = \texttt{2}$. For complete abstractions, matching corresponds to being able to use the abstraction to rewrite this subtree, resulting in compression. We refer to the set of subtrees at which a complete or partial abstraction matches as the set of \textit{match locations}. For example, the three match locations of  \texttt{(+ 3 (* $\alpha$ $\beta$))} are the subtree \texttt{(+ 3 (* (+ 2 4) 2))} in the first program, \texttt{(+ 3 (* 4 (+ 3 x)))} in the second program, and \texttt{(+ 3 (* x (+ 2 1)))} in the third program.

In traditional top-down synthesis, a branch of search can be safely pruned by proving that a program satisfying the specification cannot exist in that branch. In CTS, we can safely prune a branch of search if we can prove that it cannot contain the optimal abstraction. One way to prove this is by computing an upper bound on the utility of abstractions in the branch, and discarding the branch if we have previously found an abstraction with a utility that is higher than this bound. An efficient and conservative way to compute this upper bound is to overapproximate the set of match locations, then upper bound the compressive utility gain from each match location.

Our key observation to compute this bound is that during search, expanding a hole in a partial abstraction yields an abstraction that is more precise and thus matches at a \textit{subset} of the locations that the original matched at. This \textit{subset observation} is depicted in Fig. \ref{fig:corpus_guided} where the refinement of the partial abstraction \texttt{(+ 3 \hole)} (blue) to the goal \texttt{(+ 3 (* $\alpha$ $\beta$))} (green) results in a larger abstraction that matches at a subset of the locations (match locations shown at the top of Fig. \ref{fig:corpus_guided}B). An important consequence of this is that the match locations of a partial abstraction serves as an over-approximation of the match locations of any abstraction in this branch of top-down search.

To upper bound the compressive utility gained by rewriting at a match location, notice that the most compressive abstraction that matches at a subtree is the constant abstraction corresponding to the subtree itself. For example, the largest possible abstraction at \texttt{(+ 3 (* (+ 2 4) 2)} is just \texttt{(+ 3 (* (+ 2 4) 2)}. In other words, the size of the subtree is a strict upper bound on how much compression can be achieved at a match location.
Thus we get the upper bound for some partial abstraction $\abshole$:
\begin{equation}\label{eq:utility_upper_bound}
    U_{\text{upperbound}}(\abshole) = \sum_{e \in \text{matches}(\abshole)}{\text{size}(e)}
\end{equation} 
Equipped with this bound, our algorithm performs a branch-and-bound style top-down search, where at each step of search it discards all partial abstractions that have utility upper bounds that are less than the utility of the best complete abstraction found so far.

Our full algorithm presented in \cref{sec:technical} has some additional complexity. It handles rewriting in the presence of variables soundly, it uses an exact utility function that accounts for additional compression gained by using the same variable in multiple places, and it handles situations where match locations overlap and preclude one another. We also introduce two other important forms of pruning while maintaining the optimality of the algorithm, and we use the upper bound as a heuristic to prioritize more promising branches of search first.

\begin{figure}[t!]
    \centering
    % trim=left botm right top
    \includegraphics[clip, trim=4cm 6cm 4cm 4cm, width=\textwidth]{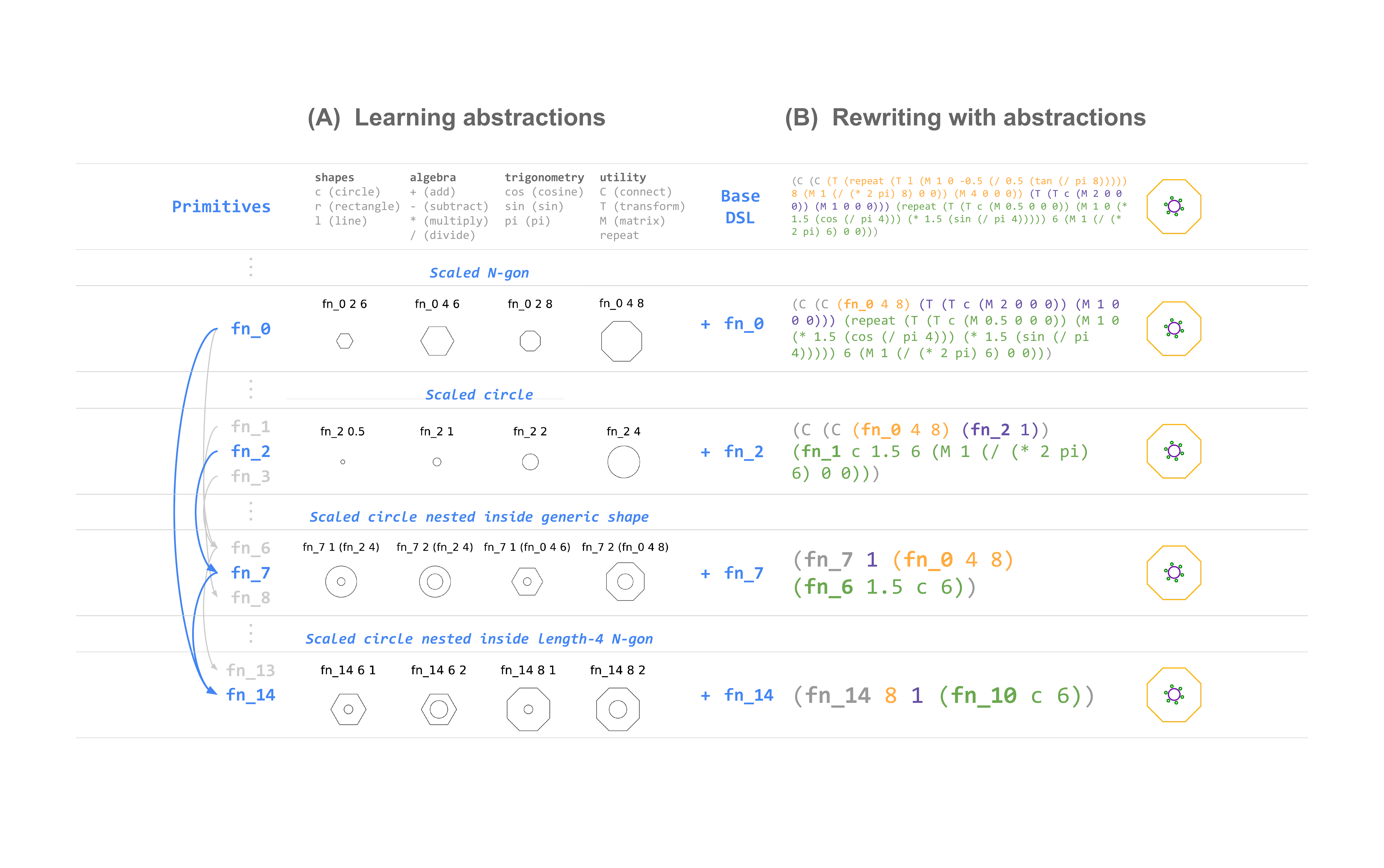}
    \caption{
    Visualization of \stitch\ library learning on the nuts-bolts subdomain from \citealt{cogsci}. (\textbf{A}) From the base DSL primitives (top row), \stitch\ iteratively discovers a series of abstractions that compress programs in the domain. Arrows demonstrate how abstractions from selected iterations build on one another to achieve increasingly higher-level behaviors. (\textbf{B}) Rewriting a single item from the domain with the cumulative benefit of discovered abstractions yields increasingly compact expressions. Colors indicate correspondence between object parts and program fragments: \textcolor{StitchOrange}{orange = outer octagon}, \textcolor{StitchGreen}{green = ring of six circles}, \textcolor{StitchPurple}{purple = inner circle}.
} \label{fig:nuts_bolts_abstractions}
% \vspace{-.75cm}
\end{figure}

\para{Building Up Abstraction Libraries.} The top-down search algorithm described above yields a single abstraction. However, we can easily run this algorithm for multiple iterations on a corpus of programs to build up an entire \textit{library of abstractions}. Fig.~\ref{fig:nuts_bolts_abstractions} illustrates the power of this kind of iterative library learning, which interleaves compression and rewriting. At each iteration, \stitch\ discovers a single abstraction that is used to rewrite the entire corpus of programs. Successive iterations therefore yield abstractions that build hierarchically on one another, achieving increasingly higher-level behaviors. As the library grows to contain richer and more complex abstractions, individual programs shrink into compact expressions.

\para{Structure of the Paper.} In the subsequent sections, we formalize our corpus-guided top-down search algorithm (\cref{sec:technical}), its application to the problem of compression (\cref{sec:framework-instantiation}), and how it may be layered on top of data structures such as version spaces (\cref{subsec:hybridapproach}). We then report experimental results (\cref{sec:experiments}) showcasing both diverse library learning settings as well as ablations of \stitch's search mechanisms. Finally, we conclude by situating \stitch\ within the landscape of related work (\cref{sec:related-work}) and future work (\cref{sec:conclusion}) in the areas of library learning and program synthesis.

\section{Corpus-Guided Top-Down Search} \label{sec:technical}

% We first provide the definitions necessary to understand our problem, then introduce our algorithm.

\subsection{Problem Setup}

% In this section we introduce necessary definitions and culminate in a statement of the utility function that the algorithm optimizes for.

In this section we provide the definitions necessary to understand our problem.

%and culminate in a statement of the utility function that the algorithm optimizes for.

\para{Grammar} Our algorithm operates on lambda-calculus expressions with variables represented through de Bruijn indices \cite{de1972lambda}; expressions come from a context-free grammar of the form $e \Coloneqq \lambda.\ e' \mid e'\ e'' \mid \$i \mid t$, where $t \in \gsym$ refers to the set of built-in primitives in the domain-specific language. For example, in an arithmetic domain $\gsym$ would include operators like \texttt{+} and constants like \texttt{3}. We say that $e'$ is a subexpression of $e$ if $e = C[e']$, where $C$ is a context as defined in the contextual semantics of \citealt{Felleisen}. An expression is \textit{closed} if it has no free variables, in which case the expression is also a \textit{program}. A set of programs $\corpus$ is  a \textit{corpus}.

When representing variables through de Bruijn indices, $\$i$ refers to the variable bound by the $i$th closest lambda above it. Thus $\lambda x.\ \lambda y.\ (x\ y)$ is represented as $\lambda.\ \lambda.\ (\$1\ \$0)$. Beta reduction with de Bruijn indices requires \textit{upshifting}: incrementing free variables in the argument when substitution recurses into a lambda, since the free variables must now point past one more lambda. Similarly, inverse beta reduction requires \textit{downshifting}.

\para{Abstraction} Given a grammar $\mathcal{G}$, we define the \textit{abstraction grammar} $\mathcal{G}_A$ as $\mathcal{G}$ extended to include \textit{abstraction variables}, denoted by Greek letters $\alpha$, $\beta$, etc. Formally, $A \Coloneqq \lambda.\ A' \mid A'\ A'' \mid \$i \mid t \mid \alpha$ where $t \in \gsym$. A term $A$ from this grammar represents the \textit{body} of an abstraction; e.g. the abstraction $\lambda \alpha.\ \lambda \beta.\ (\texttt{+}\ \alpha\ \beta)$ is simply represented by the term $(\texttt{+}\ \alpha\ \beta)$ from the language of $\mathcal{G}_A$.

\para{Partial Abstraction} A \textit{partial abstraction} $\abshole$ is an abstraction that can additionally include \textit{holes}. A \textit{hole}, denoted by $\hole$, represents an unfinished subtree of the abstraction. Thus,  $\abshole \Coloneqq A \mid \hole \mid \lambda.\  \abshole' \mid \abshole'\  \abshole''$. Any abstraction is thus also a partial abstraction. Given a grammar $\mathcal{G}$, the grammar of partial abstractions is denoted by $\mathcal{G}_{\abshole}$. Each hole in a partial abstraction $A$ can be referred to by a \textit{unique} index, which can be explicitly written as $\hole_i$.

\newcommand{\lambdaU}{\textsc{LambdaUnify}}

\para{Lambda-Aware Unification}
We introduce \textit{lambda-aware unification} as a modification of traditional unification adapted to our algorithm. $\lambdaU(A,e)$ returns a mapping (if one exists) from abstraction variables and holes to expressions $[\alpha_i \to e_i',\ \ldots,\  \hole_j \to u_j',\ \ldots]$ such that
\begin{equation}\label{eqn:beta-red}
    (\lambda \alpha_i.\ \ldots\ \lambda \hole_j.\ \ldots A)\ e_i'\ \ldots\ u_j'\ \ldots = e
\end{equation}
through beta reduction. A key difference from traditional unification is that the expression $e$ may be deep inside a program written using de Bruijn indices, so \lambdaU\ must perform some index arithmetic in order to generate its output mapping; in particular, raising a subtree out of a lambda during this inverse beta reduction requires downshifting variables in it. 

The definition of \lambdaU{} is presented in \cref{fig:lambda_unify} (left).  In \textsc{U-App}, $\textsc{merge}(l_1,l_2)$ merges two $[ \alpha_i \to e_i ]$ mappings to create a new mapping that includes all bindings from $l_1$ and $l_2$, and fails if the same abstraction variable $\alpha$ maps to different expressions in $l_1$ and $l_2$. The rule \textsc{U-Same} applies only when the abstraction argument to \lambdaU{} is an expression (i.e. it contains no holes or abstraction variables) and is syntactically identical to the expression it is being unified with. In \textsc{U-Lam}, \textsc{DownshiftAll} returns a new mapping with all abstracted expressions $e$ and hole expressions $u$ downshifted using the $\downarrow$ operator presented in \cref{fig:lambda_unify} (right). The $\downarrow$ operator has been modified from traditional downshifting to allow for partial abstractions that contain holes $\hole_j$ to break the rules of variable binding, because holes represent unfinished subtrees of the abstraction.

In particular, a new syntactic form "$\&i$" is used to represent a $\$i$ variable that has been downshifted further than traditional downshifting would permit. $\&i$ variables are created by $\downarrow$ when a traditional downshift would otherwise convert a free variable to a (different, incorrect) bound variable. These variables represent references to lambdas present within the body of the abstraction and are allowed in expressions $u_j'$ bound to holes but not expressions $e_i'$ bound to abstraction variables.

To account for $\&i$ variables, the beta reduction used in \cref{eqn:beta-red}, which is defined in \cref{fig:subst-upshift}, uses a modified upshift operator $\uparrow$ defined to be an inverse to the $\downarrow$ operator. With this modification of beta reduction, any $\&i$ variables with negative indices will ultimately be shifted back to positive indices during reduction. Furthermore, the beta reduction procedure is equivalent to traditional beta reduction when there are no holes and thus no $\&i$ variables.

\begin{figure}
    \begin{subfigure}{.45\textwidth}

    \footnotesize
    \centering
    \begin{mathpar}
    \inferrule*[left=U-AbsVar]
        { }
        {\textsc{LambdaUnify}(\alpha, e) \rightsquigarrow [ \alpha \to e ]}
    \and
    \inferrule*[left=U-Hole]
        { }
        {\textsc{LambdaUnify}(\hole_i, e) \rightsquigarrow [ \hole_i \to e ]}
    \and
    \inferrule*[left=U-App]
        {  \textsc{LambdaUnify}(A_1,e_1) \rightsquigarrow l_1 \\ \textsc{LambdaUnify}(A_2,e_2) \rightsquigarrow l_2 \\ l = \textsc{merge}(l_1,l_2) }
        {\textsc{LambdaUnify}((A_1\ A_2),(e_1\ e_2)) \rightsquigarrow l}
    \and
    \inferrule*[left=U-Lam]
        {
            \textsc{LambdaUnify}(A, e) \rightsquigarrow l' \\
            % l = [ \alpha \to \downarrow e \mid (\alpha \to e) \in l' ]
            l = \textsc{DownshiftAll}(l')
        }
        {\textsc{LambdaUnify}((\lambda.\ A),(\lambda.\ e)) \rightsquigarrow l}
    \and
    \inferrule*[left=U-Same]
        { }
        {\textsc{LambdaUnify}(e,e) \rightsquigarrow [ ]}
    \end{mathpar}
    %\caption{}
    %\label{fig:lambda_unify}
    \end{subfigure}
    \begin{subfigure}{.45\textwidth}
    \centering \footnotesize
    \begin{align*}
    &\textsc{DownshiftAll}([\alpha_i \to e_i',\ \hole_j \to u_j',\ ...]) \\&\ \ \ \ = [\alpha_i \to\ \downarrow_0 e_i',\ \hole_j \to\ \downarrow_0 u_j',\ ...]\\ \\
    &\downarrow_d \lambda.b = \lambda.\downarrow_{d+1} b \\ 
    &\downarrow_d (f x) = (\downarrow_d f\ \downarrow_d x) \\
    &\downarrow_d \$i = \begin{cases}
        \$i, & \text{if } i < d \\
        \$(i-1), & \text{if } i > d \\
        \&(i-1), & \text{if } i = d \\
        % \&(i-1), & \text{if } d \leq i < d+1 \\
    \end{cases}\\ 
    &\downarrow_d \&i = \&(i-1),\\
    &\downarrow_d t = t,\ \text{if $t$ is a primitive (i.e. $t \in \gsym$) } \\ 
    \end{align*}
    %\label{fig:downshift}
    \end{subfigure}
    \caption{(\textbf{Left}) Inference rules for lambda-aware unification. (\textbf{Right}) Definition of \textsc{DownshiftAll}.}
    \label{fig:lambda_unify}

\end{figure}

\begin{figure}
    \begin{subfigure}{.45\textwidth}
    \centering \footnotesize
    \begin{align*} 
    &  (\lambda \alpha_i.\ \ldots\ \lambda \hole_j.\ \ldots A)\ e_i'\ \ldots\ u_j'\ \ldots \\ & \ \ \ \ = [\alpha_i \to e_i',\ \hole_j \to u_j',\ ...] \subst A\\\\
    &l \subst \lambda.b = \lambda. \textsc{UpshiftAll}(l) \subst b \\
    &l \subst (f x) = (l \subst f) (l \subst x) \\
    &l \subst \$i = \$i \\
    &l \subst \alpha_i = l[\alpha_i] \\
    &l \subst \hole_i = l[\hole_i] \\
    &l \subst t = t,\ \text{for $t \in \gsym$} \\
    \end{align*}
    \end{subfigure}
    \begin{subfigure}{.45\textwidth}
    \centering \footnotesize
    \begin{align*} 
    &\textsc{UpshiftAll}([\alpha_i \to e_i',\ \hole_j \to u_j',\ ...]) \\&\ \ \ \ = [\alpha_i \to\ \uparrow_0 e_i',\ \hole_j \to\ \uparrow_0 u_j',\ ...]\\\\
    &\uparrow_d \lambda.b = \lambda.\uparrow_{d+1} b \\ 
    &\uparrow_d (f x) = (\uparrow_d f\ \uparrow_d x) \\
    &\uparrow_d \$i = \begin{cases}
        \$i, & \text{if } i < d \\
        \$(i+1), & \text{if } i \geq d \\
    \end{cases}\\ 
    &\uparrow \&i = \begin{cases}
        \&(i+1), & \text{if } i + 1 \neq d \\
        \$(i+1), & \text{if } i + 1 = d \\
    \end{cases}\\
    &\uparrow_d t = t,\ \text{for $t \in \gsym$} \\ 
    \end{align*}
    \end{subfigure}
    \caption{(\textbf{Left}) Definition of modified beta reduction and substitution ($\subst$). (\textbf{Right}) Definition of \textsc{UpshiftAll}.}
    \label{fig:subst-upshift}
\end{figure}

% \cref{appendix:lambdaunify-proof}
We provide a proof of the correctness of \lambdaU\ with respect to \cref{eqn:beta-red} in  Appendix B, as well as a Coq proof of the correctness in \texttt{stitch.v} in the supplemental material. However, to understand the key idea behind the proof we focus here on an example illustrating the different cases involving the $\downarrow$ operator and $\&i$ indices. Consider what happens when you run
\begin{equation}\label{eqn:example1}
\lambdaU(\lambda.\ f\ \hole_0,\ \lambda.\ e)
\end{equation}
By \cref{eqn:beta-red}, the goal is to produce a mapping $[\hole_0 \rightarrow u'_0]$
such than when replacing $\hole_0$ with $u'_0$ it 
produces $(\lambda.\ e)$. In other words, $(\lambda.\ \lambda.\ f\ \$1)\ u'_0 = \lambda.\ e$.
Now, suppose 
\begin{equation} \label{eqn:example1b}
\lambdaU(f~ \hole_0,\ e)\rightsquigarrow [\hole_0 \rightarrow \lambda.\ \$i]
\end{equation}
This means that $(\lambda.\ f\ \$0)\ (\lambda.\ \$i) = e$ so $e = f\ (\lambda.\ \$i)$. There are three possibilities for $i$ that need to be considered, corresponding to the 3 cases in the definition of $\downarrow_d \$i$. In the first case, when $i=0$ and thus $e = f\ (\lambda.\ \$0)$, it means that $\$i$ in \cref{eqn:example1b} is bound to the lambda in the return mapping. In this case, \textsc{DownshiftAll} in \textsc{U-Lam} should not do anything because if we can replace $\hole_0$ with $(\lambda.\ \$0)$ in $(f\ \hole_0)$ to produce $e$, the same replacement in $(\lambda.\ f\ \hole_0)$ will produce $(\lambda.\ e)$. 

The second case is when $i \geq 2$; this means that 
 $e$ has some variables that were defined outside of it and they have been captured by $\hole_0$. For example, if $i=2$ and thus $e= f\ (\lambda.\ \$2)$, the solution to \cref{eqn:example1} is $[\hole_0 \rightarrow \lambda.\ \$1]$, since 
$(\lambda.\ \lambda.\ f\ \$1) (\lambda.\ \$1) = \lambda.\ e$;
in other words, when computing 
$\lambdaU(\lambda.~ f~ \hole_0,\ \lambda.\ e)$ from 
\cref{eqn:example1b}, the $\$i$ in the return mapping has to be downshifted to account for the fact that it will be substituted inside an additional lambda. 

The third case, when $i=1$, is more problematic. In this case, we have that $e = f\ (\lambda.\ \$1)$. This creates a problem, since there is no $i$ such that 
$(\lambda.\ \lambda.\ f\ \$1) (\lambda.\ \$i) = \lambda.\ e$; the $\$i$ needs to refer to the lambda surrounding $e$. The downshift operator addresses this through the special treatment of the $\&i$ index. 

\lambdaU\ relates to prior work on unification modulo binders \cite{Huet75,Miller91,Miller92,Dowek95,Dowek96} but focuses on a more efficiently solvable syntax-driven subset of the more general unification problems tackled in this prior work. We give a more detailed comparison to this prior work in \cref{sec:related-work}.

\para{Match Locations} The set of \textit{match locations} of a partial abstraction $\abshole$ in a corpus $\corpus$, denoted $\matches{\abshole}$, is the maximal set of context-expression pairs $\{ (C_1,e_1), (C_2,e_2),\ \ldots, (C_k,e_k) \}$ such that $\forall k.\ p = C_k[e_k]$ and $p \in \corpus$ and $\textsc{LambdaUnify}(e_k,\abshole)$ succeeds. We discard locations where the mapping produced by \textsc{LambdaUnify} has \&i indices in expressions bound by abstraction variables, while we allow them in expressions bound by holes as holes are allowed to violate variable binding rules. In Section \ref{sec:algorithm} we explain how we maintain the set of matches incrementally.

\para{Rewrite Strategies} A corpus $\corpus$ can be \textit{rewritten} to use a complete abstraction $A$ as follows. We introduce a new terminal symbol $t_A$ into the symbol grammar $\gsym$ to represent the abstraction, and consider it semantically equivalent to $(\lambda \alpha_0.\ ...\ \lambda \alpha_k.\ A)$. We then re-express $\corpus$ in terms of $t_A$ as follows. We can replace the match location $(C,e) \in \matches{A}$ with $t_A$ applied with the argument assignments $[ \alpha_i \to e_i ] = \textsc{LambdaUnify}(A,e)$ to yield a semantically equivalent expression, $C[(t_A\ e_0\ ...\ e_k)]$. Note that since this is a complete abstraction, it contains no holes.

One complication is that rewriting at one match location may \textit{preclude} rewriting at another match location if they overlap with each other. A \textit{rewrite strategy} $\rewritestrategy$ is a procedure for selecting a subset of $\matches{A}$ to rewrite at, in which no match precludes another match. We refer to this set as $\rewritelocations{A}$. We refer to rewriting a corpus $\corpus$ under an abstraction $A$ with rewrite strategy $\rewritestrategy$ by $\rewrite{A}$. We refer to rewriting at only a single particular match location $m$ with abstraction $A$ by $\rewriteone{m}{A}$.

\para{Utility} CTS optimizes a user-defined utility function $\util{A}$ which scores an abstraction $A$ given a corpus $\corpus$  and rewrite strategy $\rewritestrategy$. There are no strict constraints on the form or properties of the utility function. While in \cref{sec:framework-instantiation} we will focus on compressive utility functions, our framework does not generally require this.

We say that a rewrite strategy is \textit{optimal} with respect to a utility function if the rewrite strategy chooses locations to rewrite at such that the utility is maximized. A na\"ive optimal rewrite strategy exhaustively checks the utility of rewriting with each of the $2^{|\matches{A}|}$ possible subsets of $\matches{A}$, however depending on the specific utility function computing an optimal or approximately optimal strategy may be significantly more computationally tractable. In \cref{sec:framework-instantiation} we detail the rewrite strategy used by \stitch, which is an optimal, linear-time rewrite strategy for compressive utility functions based on dynamic programming.

\subsection{Algorithm}\label{sec:algorithm}

Given a corpus $\corpus$, rewrite strategy $\rewritestrategy$, and utility function $\util{A}$, the objective of CTS is to find the abstraction $A$ that maximizes the utility $\util{A}$. CTS takes a branch-and-bound approach \cite{Land1960,morrison2016} to this problem, as described in this section.

\para{Expansion} We construct the body of a partial abstraction $\abshole$ through a series of top-down \textit{expansions} starting from the trivial partial abstraction $\hole$. Given a partial abstraction $\abshole$ we can \textit{expand} a hole in this partial abstraction using any production rule from the partial abstraction grammar $\mathcal{G}_{\abshole}$, yielding a new abstraction $\abshole'$. We denote this expansion by $\abshole 
\to \abshole'$. In Figure \ref{fig:corpus_guided}A our overall top-down search is depicted as a series of expansions. We say that a complete abstraction $A$ can be \textit{derived} from a partial abstraction $\abshole$, denoted $\abshole \to^* A$, if $A$ is in the transitive reflexive closure of the expansion operation, i.e. there exists a sequence of expansions from $\abshole$ to $A$. Any abstraction can be derived from $\hole$.

\para{A Na\"ive Approach} The goal of our algorithm is to find the maximum-utility abstraction. One simple, inefficient approach to this is to simply enumerate the \textit{entire} space of abstractions through top down synthesis and return the one with the highest utility. This na\"ive approach maintains a queue of partial abstractions initialized with just $\hole$. At each step of the algorithm it pops an abstraction from the queue, chooses a hole in it and expands that hole using each possible production rule. Whenever an expansion produces a new partial abstraction it pushes it to the queue, and when it produces a complete abstraction it calculates its utility and updates the best abstraction found so far. Since an abstraction cannot match a program that is smaller than it, the algorithm can stop expanding when all expansions would lead to abstractions larger than the largest program in the corpus. This algorithm will enumerate all possible abstractions exactly once each.

\para{Introducing Strict Dominance Pruning} While the na\"ive approach will find the optimal abstraction, it is extremely inefficient. We can improve on this using an idea core to branch-and-bound algorithms: \textit{pruning}. We allow the algorithm to choose to prune a partial abstraction instead of expanding it, in which case it simply discards the abstraction and chooses another from the queue. Of course, to maintain optimality we must be certain that we \textit{never} prune the branch of search containing the optimal abstraction. To formalize safe pruning we will use the notion of \textit{strict dominance} from branch-and-bound literature \cite{morrison2016,Ibaraki1977,Chu2015}, which we explain below in terms of a notion of \textit{covering}.

A complete abstraction $A''$ \textit{covers} another complete abstraction $A'$, written $\text{covers}(A'',A')$, if the utility of $A''$ is greater than that of $A'$. A partial abstraction $\abshole'$ is \textit{strictly dominated by} another partial abstraction $\abshole''$ if and only if every complete abstraction that is derivable from $\abshole'$ is covered by a complete abstraction that is derivable from $\abshole''$.

Formally, $\abshole''$ strictly dominates $\abshole'$ iff $\forall A'.\ \abshole' \to^* A' \implies \exists A''.\ \abshole'' \to^* A'' \land \text{covers}(A'',A')$.

Equipped with this formalism, we claim that it is safe to prune an abstraction $\abshole'$ if we know that there exists some $\abshole''$ which strictly dominates it. Note that $\abshole''$ does not necessarily contain the optimum and may even be pruned in the search if it is strictly dominated by another abstraction. Knowing the existence of $\abshole''$ is enough to enable pruning, regardless of whether it has been pruned or has not yet been enumerated during search.

\begin{lemma}\label{lemma:subsumption-optimal}
 Na\"ive search augmented with strict dominance pruning finds the optimal abstraction.
\end{lemma}
\begin{proof}
We proceed with a proof by induction. We seek to prove the equivalent statement that the optimum is never pruned and thus it will be enumerated by the search.

Our inductive hypothesis is that the optimum has not yet been pruned. In the base case this is trivially true, since no pruning has taken place yet. In the inductive step we must prove that a step of pruning maintains the validity of the inductive hypothesis.
Recall that a step of pruning will discard some partial abstraction $\abshole'$ which is strictly dominated by another abstraction $\abshole''$.
Suppose then for the sake of contradiction that we did prune the optimum in this step; then the optimum must have been derivable from $\abshole'$. By strict dominance we know that all abstractions derivable from $\abshole'$ must be covered by some abstraction derivable from $\abshole''$. Hence, there must exist some abstraction that covers the optimum. However, since the utility of the optimum is greater than or equal to that of all other abstractions, there is by definition \textit{no} such abstraction; we have thus arrived at a contradiction. 
\vspace{-2pt}
\end{proof}
% \vspace{-5pt}
This lemma ensures the safety of composing together pruning strategies, as long as they all are forms of strict dominance pruning, since no instance of strict dominance pruning will remove the optimum. Determining which partial abstractions are strictly dominated by which others is specific to the utility function being used, and in Section \ref{sec:framework-instantiation} we identify two instances of strict dominance used when instantiating this framework for a compression-based utility.

\para{Upper Bound Pruning}
We further improve this algorithm to employ upper bound based pruning, in which we bound the maximum utility that can be obtained in a branch of search, and discard the branch if we have previously found a complete abstraction with higher utility than this bound. This is the most common form of pruning used in branch-and-bound algorithms (see \citealt{morrison2016} Section 5.1 for a review). Our algorithm is generic over the upper bound function $\upperbound{\abshole}$ which upper bounds the utility of \textit{any complete abstraction} $A$ that can be derived from  $\abshole$.

Formally, the bound must satisfy $\forall A.\ \abshole \to^* A \implies \util{A} \leq \upperbound{\abshole}$. While this could trivially be satisfied with $\upperbound{\abshole} = \infty$ for any choice of a utility function, a tighter bound will allow for more pruning. The soundness of upper bound pruning in maintaining the optimality of the solution trivially follows from the fact that the pruned branch only contains complete abstractions that are at most as high-utility as the best abstraction found so far.

While we defer to \cref{sec:framework-instantiation} to construct the upper bound used by \stitch, it is worth mentioning a key insight that helps in constructing tight bounds for many utility functions. When bounding the utility of $\abshole$, it is useful to bound the set of possible locations where rewrites can occur for any abstraction $A$ derived from $\abshole$. The following lemma provides such a bound:

\begin{lemma}\label{lemma:matches-subset-rewrites}
$\matches{\abshole}$ is an upper bound on the set of locations where rewrites can occur in any $A$ derived from $\abshole$.
\end{lemma}

This follows from the fact that as partial abstractions are expanded they become more precise and thus match at a subset of the locations. Since rewriting only happens at a subset of match locations, $\matches{\abshole}$ serves as a bound on the set of locations where rewrites can occur.

An important consequence of \cref{lemma:matches-subset-rewrites} is that if a partial abstraction matches at zero locations, then all abstractions derived from it will have zero rewriting locations, so no rewriting can occur. Such branches can therefore be safely pruned.

\para{Improving the Search Order}

Finally, without sacrificing the optimality of this algorithm we can heuristically guide the order in which the space of abstractions is explored. The queue of partial abstractions used in top-down search can be replaced with a priority queue ordered by the upper bound. This way, the algorithm will first explore more promising, higher-bound branches of search, narrowing in on the optimal abstraction more quickly. In branch-and-bound literature this choice of a search order that uses a sound upper bound is sometimes referred to as \textit{best-bound} search and makes the algorithm a form of A$^*$ search \cite{Hart1968} since the upper bound is an admissible heuristic \cite{morrison2016}.

In preliminary experiments on a subset of benchmarks we found that search order had little effect on the overall runtime of the algorithm, but the best-bound ordering was moderately helpful in more quickly narrowing in on the optimal solution (i.e. useful when running the algorithm with a limited time budget), so this is used for all experiments.

\para{Efficient Incremental Matching} 
When expanding an abstraction $\abshole$ to a new abstraction $\abshole'$, it is easy to compute $\matches{\abshole'}$ since we know it is a subset of $\matches{\abshole}$. Thus, there is no need to perform matching from scratch against every subtree in the corpus to compute the match locations. Instead, we can simply inspect the relevant subtree at each match location of the original abstraction to see which match locations will be preserved by a given expansion. In fact, except when expanding into an abstraction variable $\alpha$, the sets of match locations obtained by different expansions of a hole will be disjoint.

When expanding to an abstraction variable $\alpha$, if $\alpha$ is an existing abstraction variable then this is a situation where the same variable is being used in more than one place, as in the \textit{square} abstraction $(\lambda \alpha.\ \texttt{*}\ \alpha\  \alpha)$. In this case we restrict the match locations to the subset of locations where within the location all instances of $\alpha$ are bound to syntactically identical subtrees.

Additionally, if the user provides a maximum arity limit, then an expansion that causes the abstraction to exceed this limit can be discarded.

\para{Algorithm Summary and Key Points}
The CTS algorithm combines all of the aforementioned optimizations in a top-down search algorithm with pruning. In summary, CTS takes as input a corpus $\corpus$, a rewrite strategy $\rewritestrategy$, a utility function $\util{A}$, and an upper bound function $\upperbound{\abshole}$, and returns the \textit{optimal abstraction} with respect to the utility function. CTS performs a top-down search over partial abstractions and prunes branches of search when partial abstractions match at zero locations or are eliminated by upper bound pruning or strict dominance pruning. CTS is made further efficient through the aforementioned \textit{search-order heuristic} and \textit{efficient incremental matching}. Importantly, none of these optimizations sacrifice the optimality of the abstraction found --- all pruning is done soundly, as discussed in the prior sections on upper bound pruning and strict dominance pruning.

We also note that CTS is amenable to parallelization without losing optimality. This can be implemented using a shared priority queue that is safely accessed by different worker threads through a locking primitive. Strict dominance pruning trivially remains sound as it doesn't depend on any global information about the state of the search. Upper bound pruning also remains sound even if workers only occasionally synchronize their best-abstraction-so-far, as this will just mean that they occasionally have weaker upper bounds which are still sound.

Since the algorithm maintains a best abstraction so-far, it can also be terminated early, making it an \textit{anytime algorithm}. Finally, to learn a library of abstractions, CTS can be run repeatedly (much like DreamCoder), adding one abstraction at a time to the library and rewriting the corpus with each abstraction as it is learned before running CTS again. A listing of the full algorithm is provided in Appendix A. %\cref{appendix:cts-full}.

\section{Applying Corpus-Guided Top-Down Search to Compression}\label{sec:framework-instantiation}

Having presented the general framework and algorithm of corpus-guided top-down search, we now instantiate this framework for optimizing a \textit{compression} metric.

\subsection{Utility}\label{sec:utility}

In compression we seek to minimize some measure of the size, or more generally the \textit{cost}, of a corpus of programs after rewriting them with a new abstraction. As in prior work \citep{ellis2021dreamcoder} we penalize large abstractions by including abstraction size in the utility. The compressive utility function for a corpus $\corpus$ and abstraction $A$ is given below.

\begin{equation}\label{eqn:utility-1}
 \util{A} \triangleq - \cost{A} + \cost{\corpus} - \cost{\rewrite{A}}
\end{equation}

Here, $\cost{\cdot}$ is a cost function of the following form:

\begin{equation} \label{eqn:cost}
    \begin{aligned}
    &\cost{\lambda.\ e'} = \costlam + \cost{e'}\\
    &\cost{e' e''} = \costapp + \cost{e'} + \text{cost}(e'')\\
    &\cost{\$i} = \costvar\\
    &\cost{t} = \costt{t},\text{ for } t \in \gsym\\
    &\cost{\alpha} = \costabs\\
    \end{aligned}
\end{equation}

where $\costlam$, $\costapp$, $\costvar$, and $\costabs$ are non-negative constants, and $\costt{t}$ is a mapping from grammar primitives to their (non-negative) costs. Finally, we introduce $\coststar{\cdot}$ as a version of $\cost{\cdot}$ where $\costabs = 0$.

We can equivalently construct the utility in \cref{eqn:utility-1} by summing over the compression gained from performing each rewrite separately, as given below.

\begin{equation}\label{eqn:util-decomposed}
 \util{A} = - \cost{A} + \sum_{e \in \rewritelocations{A}}{\cost{e} - \cost{\rewriteone{e}{A}}}
\end{equation}

By reasoning about the way that rewriting transforms a program, this utility can be broken down even further:
\begin{equation}\label{eqn:uglobal}
% \begin{aligned}
 \util{A} = - \cost{A} + \sum_{e \in \rewritelocations{A}}{U_{local}(A,e)}
\end{equation}
% \\
% \\
\begin{equation}
 U_{local}(A,e) = \!\!  \underbrace{\coststar{A}
    \vphantom{\sum_{\alpha \in \text{a}(A)}} % phantom
    }_{\text{abstraction size}} -
    \underbrace{( \costt{t_A} + \costapp \cdot \text{arity}(A))
    \vphantom{\sum_{\alpha \in \text{a}(A)}} % phantom
    }_{\text{application utility (negative)}}
    + \underbrace{\sum_{[\alpha \to e'] \in \text{args}(A,e)}{
    \!\!\!\!\!\!\!\!\!\!\!
    (\text{usages}(\alpha) - 1) \cdot \text{cost}(e')}}_{\text{multiuse utility}}
% \end{aligned}
\label{eqn:ulocal}
\end{equation}

where $t_A$ is the new primitive corresponding to abstraction $A$. This form of the utility function can be efficiently computed without explicitly performing any rewrites; and it sheds light on the different sources of compression. The three main terms in this expression are (1) the \textit{abstraction size} that comes from the shared structure that is removed, (2) the negative \textit{application utility} that comes from introducing the new primitive and lambda calculus \texttt{app} nonterminals to apply it to each argument, and (3) the \textit{multiuse utility} which comes from only needing to pass in a single copy of an argument that might be use in more than one place in the body. We emphasize that this form of the utility function is equivalent to the original definition based on rewriting given in \cref{eqn:utility-1}.

\subsection{Upper Bounding the Utility}

We seek an upper bound function $\upperbound{\abshole}$ such that for any $A$ derived from $\abshole$, $\upperbound{\abshole} \geq \util{A}$. We begin from the decomposition of the utility function given in \cref{eqn:util-decomposed}. Since costs are always non-negative, we can upper bound this by dropping the $-\text{cost}(A)$ term as well as the negative term within the sum:

\begin{equation}
    \util{A} \leq \sum_{e \in \rewritelocations{A}}{\text{cost}(e)}
\end{equation}

Intuitively, dropping the negative term from the sum is equivalent to assuming we compressed the cost of this location all the way down to cost 0. We can also bound $\rewritelocations{A}$ as $\matches{\abshole}$ using \cref{lemma:matches-subset-rewrites}, yielding our final upper bound in terms of $\abshole$:

\begin{equation}
    \upperbound{\abshole} \triangleq \sum_{e \in \matches{\abshole}}{\text{cost}(e)}
\end{equation}

\subsection{Strict Dominance Pruning}\label{subsec:strict-dom-compression}
We identify two forms of strict dominance pruning compatible with a compressive utility. The first is \textit{redundant argument elimination}: a partial abstraction can be dropped if it has two abstraction variables that always take the same argument as each other across all match locations, for example if it had variables $\alpha$ and $\beta$ and $\alpha = \beta = \texttt{(+ 3 5)}$ at one location, $\alpha = \beta = \texttt{2}$ at another location, and so on for all match locations. This abstraction would then be strictly dominated by the abstraction that doesn't take $\beta$ as an argument and instead reuses $\alpha$ in place of $\beta$, and can hence be eliminated. Thus an abstraction like \texttt{(+ (* $\alpha$ $\beta$) \hole)} is strictly dominated by an abstraction like \texttt{(+ (* $\alpha$ $\alpha$) \hole)} if $\alpha = \beta$ across all match locations in the former abstraction.

The second is \textit{argument capture}: when a partial abstraction takes the same argument for at least one abstraction variable across all match locations, this abstraction can be discarded. This is because every abstraction derivable is covered by another abstraction which is identical except it has the argument inlined into the body. For example, if in the abstraction \texttt{(+ (* $\alpha$ $\beta$) ??)} the abstraction variable $\alpha$ takes the same argument $\alpha=$\texttt{(+ 3 5)} across every match location, there is a strictly dominating partial abstraction that simply has \texttt{(+ 3 5)} inlined in its body: \texttt{(+ (* (+ 3 5) $\beta$) ??)}. Note that argument capture does not apply when the argument contains a free variable, as inlining would result in an invalid abstraction.

A close reader might notice that if the number of times $\alpha$ appears in $A$ is greater than the number of rewrite locations, the abstraction with argument capture applied, $A'$, will have slightly lower utility. Employing this pruning rule in our search means that we will find the optimal abstraction subject to the constraint that all possible argument captures have occurred. This utility difference from the optimal abstraction without argument capture is bounded by at most $\cost{A'}$, since the difference comes from the contribution of the size of the abstraction body itself to the utility.
\footnote{
Specifically, for some $[\alpha \to e]$, the abstraction without argument capture is higher in utility by $- (\cost{e} + \costapp) * |\rewritelocations{A}| + (\cost{e} - \costabs) * \text{usages}(\alpha)$. For a cost function where $\costapp < \costabs$, this is positive when the argument is used more times than the abstraction itself is used.
}
In preliminary experiments across all of our experimental datasets, we never find this edge case to change which abstraction is chosen as optimal.

\subsection{Rewrite Strategy}

\stitch\ employs a linear-time, optimal rewrite strategy for a compression objective. The goal of an optimal rewrite strategy is to efficiently select the optimal subset of $\matches{A}$ to perform rewrites at in order to maximize the utility. 

The main challenge is that when match locations overlap, the strategy must decide which of the two to accept. For example, consider the program \texttt{(foo (foo (foo bar))} and the abstraction \texttt{$t_A = \lambda \alpha.\ $(foo (foo $\alpha$))}. This abstraction matches at the root of the program with $\alpha = $ \texttt{(foo bar)}, resulting in the rewritten program \texttt{($t_A$ (foo bar))}. However, though the abstraction also matched at the subtree \texttt{(foo (foo bar)} with $\alpha = $ \texttt{bar} in the original program, this match location is no longer present in the rewritten program, so only one of the two locations can be chosen by the rewrite strategy.

Our approach is a bottom-up dynamic programming algorithm, which begins at the leaves of the program and proceeds upwards. At each node $e$, we compute the cumulative utility so far if we reject a rewrite here ($\text{util}_R$), if we accept a rewrite here ($\text{util}_A$), or if we choose the better of the two options ($\text{util}^*$):

\begin{equation}
    \begin{aligned}
    &\text{util}_R[e] = \sum_{e' \in children(e)}{\text{util}^*[e']}\\
    &\text{util}_A[e] = \left\{
        \begin{array}{ll}
        0 & \text{if } e \notin \matches{A}\\
        U_{local}(A,e) + \sum\limits_{e' \in \text{args}(A,e)}{\text{util}^*[e']} & \text{otherwise}\\
        \end{array}\right. \\ 
    &\text{util}^*[e] = \text{max}(\text{util}_R[e],\text{util}_A[e])
    \end{aligned}\label{eqn:dynamic-programming}
\end{equation}

where $U_{local}(A,e)$ is the utility gained from a single rewrite as defined in \cref{eqn:ulocal}. After calculating these quantities the rewrite strategy can start from the program root and recurse down the tree, rewriting at each node where rewriting is optimal (i.e. $\text{util}_A > \text{util}_R$) and then recursing into its arguments after rewriting. Since all arguments were originally subtrees in the program, these quantities will have been calculated for all of them, with the caveat that their de Bruijn indices may have been shifted. Shifting indices of a subtree does not change whether an abstraction can match there nor does it change the utility gained from using that abstraction, so this simply requires some extra bookkeeping to track.

This algorithm is optimal by a simple inductive argument: at each step of the dynamic programming problem, we can assume that we know the cumulative utility of all children and potential arguments at this node, so we can use \cref{eqn:dynamic-programming} to calculate the cumulative utility of either rejecting or accepting the rewrite at this node.

\section{Combining corpus-guided top-down search with deductive approaches}
\label{subsec:hybridapproach}

In complex domains, assembling good libraries may involve more than just finding matching code-templates; sometimes, some initial refactoring is necessary to expose common structure.
For example, consider learning the abstraction $\lambda \alpha.\ \texttt{(* 2 }\alpha\texttt{)}$ for doubling integers, given the expressions \texttt{(* 2 8)}, \texttt{(* 7 2)}, and  \texttt{(right-shift 3 1)}.
This is only possible if the system can use the commutativity of multiplication to rewrite \texttt{(* 7 2)} into \texttt{(* 2 7)}, and bitvector properties to rewrite \texttt{(right-shift 3 1)} into \texttt{(* 2 3)}; such rewrites are not natively supported by CTS.

In this section, we discuss how CTS can be combined with refactoring systems based on deductive rewrites to increase its expressivity further.
The core idea is simple: run a rewrite system on the corpus to produce a set of refactorings of the corpus in a version space, then run CTS over the resulting data structure.\footnote{Note that while we have previously presented CTS as operating over syntax trees, the core notions of upper bounds and matching that CTS operates on are not restricted to program trees.} 
Intuitively, this will lead to improved performance compared with a purely deductive approach since the cost of rewriting grows exponentially with the number iterations of rewrites that are applied. This is a problem for fully deductive approaches like \citealt{ellis2021dreamcoder} because extracting the abstractions often requires a long chain of rewrites, especially for higher-arity abstractions. However, a small number of rewrites is typically sufficient to expose the underlying commonalities, as it was in the example above; performing only a handful of rewrites and then using CTS to actually extract the library thus avoids the exponential blow-up of computing several rewrites in sequence, while still benefiting from the increase in expressivity afforded by the rewrites.

\para{Version Spaces} To illustrate this hybrid approach, we combine CTS with \emph{version spaces}~\cite{lau2003programming,polozov2015flashmeta} representing sets of programs semantically equivalent under the $\beta$-inversion rewrite. Version spaces are represented as terms from a grammar obtained by extending the grammar of expressions with the union ($\uplus$) operator that represents a set of equivalent expressions:
$v \Coloneqq \varnothing\mid v'\ \uplus\ v''\mid \lambda.\ v' \mid v'\ v'' \mid \$i \mid t$.
We define a denotation operator, $\eval{v}$, mapping a version space $v$ to a set of terms:
for the union operator, $\eval{v\uplus v'}=\eval{v}\cup\eval{v'}$;
for applications, $\eval{v\ v'}=\left\{ e\ e'\,:\,\forall e,e'\in \eval{v}\times\eval{v'}\right\}$;
for lambda abstractions, $\eval{\lambda v}=\left\{ \lambda e\ \,:\,\forall e\in \eval{v}\right\}$;
for the empty set, $\eval{\varnothing}=\varnothing$;
and for de Bruijn indices and terminals, $\eval{v}=v$.
\citealt{ellis2021dreamcoder} gives a procedure called $I\beta(e)$, which take as input an expression $e$ and outputs a version space $v$ inverting one step of $\beta$-reduction: that is, $e'\to^\beta e$ iff $e'\in \eval{I\beta(e)}$.

\begin{figure}
\centering\footnotesize
    \begin{mathpar}
    \inferrule*[left=RU-AbsVar]
        { }
        {\textsc{RewriteUnify}(\alpha,v) \rightsquigarrow [ \alpha \to v ]}
    \and
    \inferrule*[left=RU-Hole]
        { }
        {\textsc{RewriteUnify}(\hole_i,v) \rightsquigarrow [ \hole_i \to v ]}
    \and
    \inferrule*[left=RU-App]
        {  \textsc{RewriteUnify}(A_1,v_1) \rightsquigarrow l_1 \\ \textsc{RewriteUnify}(A_2,v_2) \rightsquigarrow l_2 \\ l = \textsc{vmerge}(l_1,l_2) }
        {\textsc{RewriteUnify}((A_1\ A_2),(v_1\ v_2)) \rightsquigarrow l}
    \and
    \inferrule*[left=RU-Lam]
        {
            \textsc{RewriteUnify}(A,v) \rightsquigarrow l' \\
            % l = [ \alpha \to \downarrow e \mid (\alpha \to e) \in l' ]
            l = \textsc{DownshiftAll}(l')
        }
        {\textsc{RewriteUnify}((\lambda.\ A),(\lambda.\ v)) \rightsquigarrow l}
    \and
    \inferrule*[left=RU-Same]
        { }
        {\textsc{RewriteUnify}(v,v) \rightsquigarrow [ ]}
    \and \inferrule*[left=RU-Union]
        {\textsc{RewriteUnify}(A,v_1) \rightsquigarrow l}
        {\textsc{RewriteUnify}(A,v_1\uplus v_2) \rightsquigarrow l}
    \\\and
    \inferrule*[left=VM-Same]
        { [\alpha\to v_1]\in l_1\qquad[\alpha\to v_2]\in l_2}
        {[\alpha\to (v_1\cap v_2)]\in \textsc{vmerge}(l_1,l_2)}
    \and
    \inferrule*[left=VM-Diff]
        { [\alpha\to v]\in l_1\\\alpha\not\in \text{args}(l_2)}
        {[\alpha\to v]\in \textsc{vmerge}(l_1,l_2)}
    \end{mathpar}
    \caption{Defining \textsc{RewriteUnify}, which takes as input an abstraction and a version space of possible refactorings, and yields multiple substitutions corresponding to all the ways that the abstraction's holes and variables can match with programs encoded by the version space.}\label{fig:rewriteunify}
\end{figure}

To run CTS on top of this rewriting system requires a generalization of \textsc{LambdaUnify}: instead of unifying an abstraction with a term, yielding a single substitution, we unify against a \emph{set of terms} (a version space), yielding a set of candidate substitutions.
The relation \textsc{RewriteUnify} (\cref{fig:rewriteunify}) accomplishes this, and equipped with this relation we can express the size of a subtree $e$ after expanding it into the version space $I\beta(e)$ and rewriting it with abstraction $A$ as:
\begin{equation}
\text{cost}(\textsc{Rewrite}(A,e))=\min_{l:\; \textsc{RewriteUnify}(A,I\beta(e)) \rightsquigarrow l}
%\text{cost}(\textsc{Rewrite}(A,e))=\min_{\substack{l:\\\textsc{RewriteUnify}(A,I\beta(e)) \rightsquigarrow l}}
  \costt{t_A} + \costapp \cdot \text{arity}(A)
    % phantom
    +\sum_{v\in \text{args}(l)} \min_{e'\in \eval{v}} \text{cost}(e')
\end{equation}

We then can approximate the utility of rewriting a corpus (based on \cref{eqn:util-decomposed}) as follows 

\begin{equation}
 \util{A} \approx -\cost{A} + \sum_{p \in \corpus}{\; \max_{e \in \text{subtrees}(p)}{ \cost{e} - \text{cost}(\textsc{Rewrite}(A,e)) } }
\end{equation}

This utility is exact when the optimal way of rewriting the corpus using $A$ has at most one rewrite location per program. This is because it considers the utility of the single best site at which to perform the rewrite instead of considering multiple simultaneous rewrites at different locations within a single program. We can bound this approximate utility for a partial abstraction by computing the approximate utility directly while treating any version space bound to a hole in $\text{args}(l)$ as having zero cost.

\para{Example: Learning map from fold.}
Consider learning the higher-order function 
\texttt{map $= \lambda\alpha.$(fold ($\lambda x.\lambda\ell.$ (cons ($\alpha$ $x$) $\ell$)))} 
from two example programs: doubling a list of numbers, expressed as  
\texttt{(fold ($\lambda x.\lambda\ell.$ (cons (+ $x$ $x$) $\ell$)))},
and decrementing a list of numbers, expressed as 
\texttt{(fold ($\lambda x.\lambda\ell.$ (cons (- $x$ 1) $\ell$)))}.
Matching these programs with the \texttt{map} function requires re-expressing them as 
\texttt{(fold ($\lambda x.\lambda\ell.$ (cons (($\lambda z.$ (+ z z)) $x$) $\ell$)))}
and 
\texttt{(fold ($\lambda x.\lambda\ell.$ (cons (($\lambda z.$ (- z 1)) $x$) $\ell$)))},
which matches the \texttt{map} function with the substitutions $\alpha = \texttt{($\lambda z.$ (+ z z))}$ and $\alpha = \texttt{($\lambda z.$ (- z 1))}$ respectively.
These two re-expressions are performed by running a single round of $\beta$-inversion as a deductive rewriting step.

While the majority of our experiments will focus on evaluating the approach laid out in \cref{sec:technical} and \cref{sec:framework-instantiation}, we implement and evaluate a prototype of integrating version spaces in this way in \cref{subsec:claim-5}.

\section{Experiments}\label{sec:experiments}

In this section, we evaluate corpus-guided top-down search for library learning. Specifically, our evaluation focuses on five hypotheses about the performance of \stitch{}: 
\begin{enumerate}
    \item \textit{\stitch\ learns libraries of comparable quality to those found by existing \textit{deductive library learning} algorithms in prior work, while requiring significantly less resources.} In \cref{subsec:claim-1} we run \stitch\ on the library learning tasks from \cite{ellis2021dreamcoder} and directly compare \stitch\ to DreamCoder, the deductive algorithm introduced in that work. We find that \stitch\ learns libraries which usually match or exceed the baseline in quality (measured via a compression metric), while improving the resource efficiency in terms of memory usage and runtime by 2 and 3-4 orders of magnitude compared to the baseline (respectively).
    
    \item \textit{\stitch\ scales to corpora of programs that contain more and longer programs than would be tractable with prior work.} In \cref{subsec:claim-2}, we evaluate \stitch's ability to learn libraries within eight graphics domains from \citealt{cogsci}, which are considerably larger and more complex than have been considered in previous work. We find that \stitch\ on average obtains a test set compression ratio of 2.55x-11.57x in 0.19s-60.16s, with a peak memory usage of 11.21MB-714.55MB in these domains. The problems are large enough to time out with the DreamCoder baseline. 
    
    \item \textit{\stitch\ degrades gracefully when resource-constrained.} In \cref{subsec:claim-3} we investigate \stitch's performance when run as an \textit{anytime} algorithm; i.e., one that can be terminated early for a best-effort result if a corpus is too large or there are limits on time or memory. We reuse the eight graphics domains from \citealt{cogsci} and find that with its heuristic guidance \stitch\ converges upon a set of high quality abstractions very early in search, doing so within 1\% of the total search time in 3 out of 8 domains and within 10\% in all except one.
    
    \item \textit{All the elements of \stitch\ matter.}  In \cref{subsec:claim-4}, we carry out an ablation study on \stitch\ and find the \textit{argument capturing} and \textit{upper bound} pruning methods are essential to its performance, while \textit{redundant argument elimination} also proves useful in certain domains. With all optimizations disabled, we find that \stitch\ cannot run in $\leq 90$ minutes and $\leq 50$GB of RAM on any of the domains from \citealt{cogsci}.
    
    \item \textit{\stitch\ is complementary to deductive rewrite-based approaches to library learning}. These prior experiments show the superior runtime performance of \stitch\ relative to deductive rewrite systems, but deductive systems have an important advantage over \stitch{}: the ability to incorporate arbitrary rewrite rules to expose more commonality among different programs and in that way discover better libraries.  Such deductive approaches are especially more apt at learning higher-order abstractions.
    In \cref{subsec:claim-5}, we give evidence that this expressivity gap can be reduced by running \stitch\ on top of a deductive rewrite system, allowing it to learn many new abstractions while still using $<$2\% of the compute time.

\end{enumerate}

For all experiments, we parameterize \stitch's $\cost{e}$ function (as defined in \cref{sec:utility}) as follows: $\costapp = \costlam = 1$, $\costvar = \costabs = \costt{t} = 100$. To avoid overfitting, DreamCoder prunes the abstractions that are only useful in programs from a single task. We add this to \stitch\ as well, treating each program as a separate task for datasets that don't divide programs into tasks.

We run all experiments on a machine with two AMD EPYC 7302 processors, 64 CPUs, and 256GB of RAM. We note however that \stitch\ itself runs exceptionally well on a more average machine. For example, on one author's laptop (ThinkPad X1 Carbon Gen 8), the experiments from \cref{subsec:claim-2} can be replicated with nearly identical runtimes, and these are the most computationally intensive experiments outside of the ablation study.

\subsection{Iterative Bootstrapped Library Learning}
\label{subsec:claim-1}

\subsubsection*{Experimental setup}
Our first experiment is designed to replicate the experiments in DreamCoder, which is the state-of-the-art in deductive library learning. DreamCoder learns libraries iteratively: the system is initialized with a low-level DSL, and then alternates between synthesizing programs (via a neurally-guided enumerative search) that solve a training corpus of inductive tasks and updating the library of abstractions available to the synthesizer. Traces from the experiments carried out by \citealt{ellis2021dreamcoder} are publicly available\footnote{\url{https://github.com/mlb2251/compression_benchmark}} and include all of the intermediate programs that were synthesized as well as the libraries learned from those programs.

In this experiment, we take these traces and evaluate \stitch\ on each instance where library learning was performed, comparing the quality of the resulting library to the original one found by DreamCoder. We also re-run DreamCoder on these same benchmarks in order to evaluate its resource usage, capturing its runtime and memory usage in the same environment as  \stitch.

The library learning algorithm in \citealt{ellis2021dreamcoder} implements a stopping criterion to determine how \textit{many} abstractions to retain on any given set of training programs. In our comparative experiments, we run the DreamCoder baseline first, and then match the number of abstractions learned by \stitch\ at each iteration to those learned by the baseline under its stopping criterion so that timing comparisons are fair. We record the total time spent both performing abstraction learning and rewriting for both \stitch\ and DreamCoder.

We replicate experiments on five distinct domains from \citealt{ellis2021dreamcoder}:
\begin{itemize}
    \item \textbf{Lists}: A functional programming domain consisting of 108 total inductive tasks.
    \item \textbf{Text}: A string editing domain in the style of FlashFill \cite{gulwani2015inductive} consisting of 128 total inductive tasks.
    \item \textbf{LOGO}: A graphics domain consisting of 80 total inductive tasks.
    \item \textbf{Towers}: A block-tower construction domain consisting of 56 total inductive tasks.
    \item \textbf{Physics}: A domain for learning equations corresponding to physical laws from observations of simulated data, consisting of 60 total inductive tasks.
\end{itemize}

\subsubsection*{Assessing library quality with a compression metric}\label{sec:metrics}

The standard \stitch{} configuration optimizes a compression metric that minimizes the size of the programs after being rewritten to use the abstraction. This is a standard metric in program synthesis, since shorter programs are frequently easier to synthesize. Optimizing against this metric is equivalent to maximizing the likelihood of the rewritten programs under a uniform PCFG.  

The DreamCoder synthesizer is more sophisticated than simple enumeration; it takes as input a learned typed bigram PCFG and leverages it to synthesize programs more efficiently. When performing compression, it optimizes against this given PCFG in order to find abstractions that will be more profitable for its specific synthesizer. For the purpose of this evaluation, however, we restrict ourselves to the uniform PCFG because the one used by DreamCoder requires programs to be in a particular normal form. 
Another aspect of DreamCoder relevant to its compression metric is that DreamCoder synthesizes multiple programs that solve the same task and then selects the abstraction that works best on \emph{some} program for a given task. This is expressed formally in the equation below. 
\begin{equation}\label{eqn:min-metric}
    \text{cost}(A)+\sum\limits_{\text{task}}\min\limits_{p \in \text{task}}\text{cost}\left( \textsc{Rewrite}(p,A) \right)%}
\end{equation}
It is trivial to implement this best-of-task metric in \stitch{}, so we use this for the comparisons with DreamCoder.

\hfill\\ \vspace{-8pt}
\begin{figure}[t]
    \centering
    \includegraphics[width=.6\textwidth]{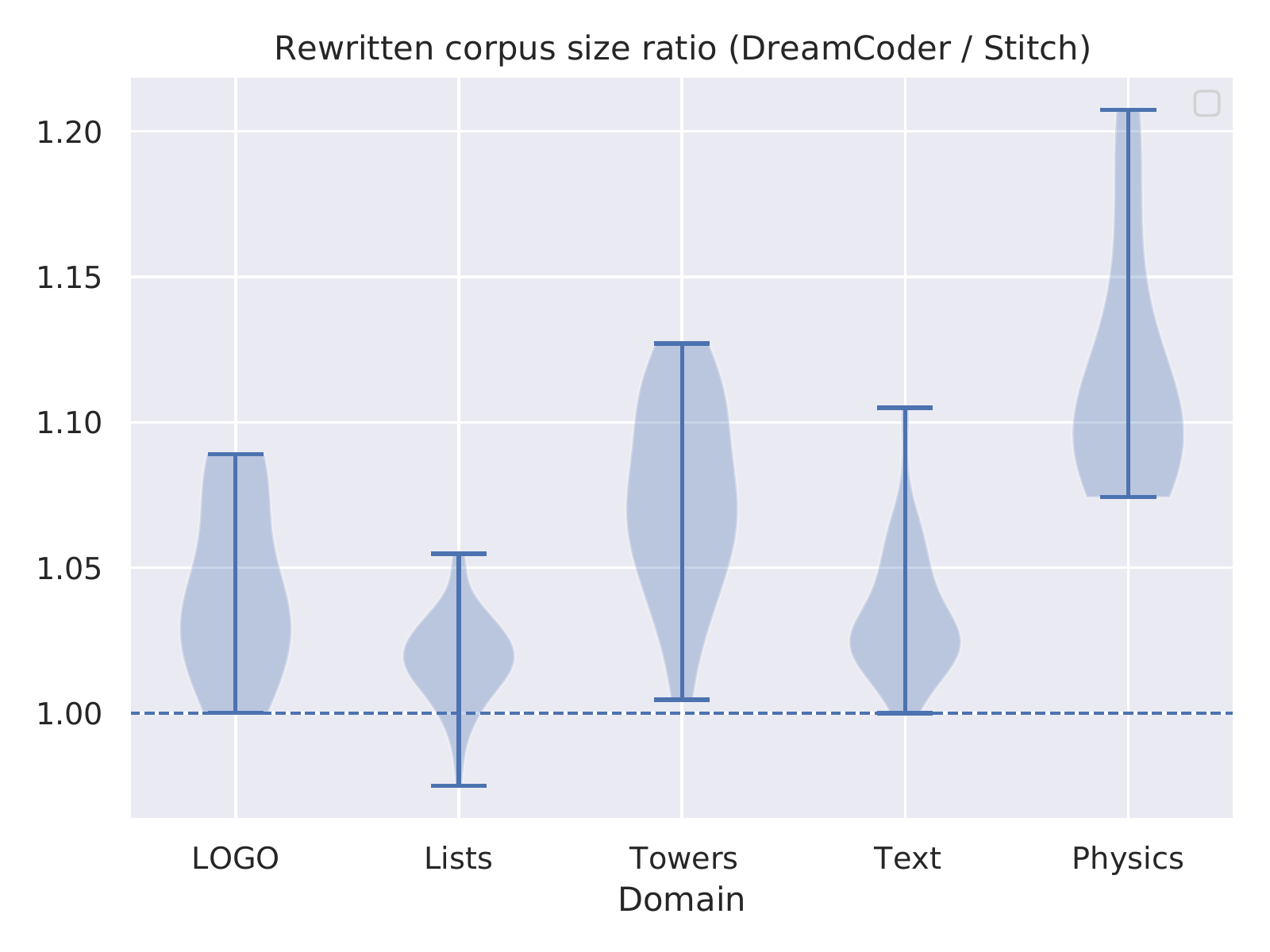}
    \caption{Compression rates obtained when running \stitch\ on the five domains considered in \ref{subsec:claim-1} relative to those of DreamCoder. Higher is better for \stitch: a ratio above 1.0 indicates that \stitch\ achieves greater compression than DreamCoder. The specific compression metric used in the ratio is given in \cref{eqn:min-metric}.}
    \label{fig:casestudy1-compression_ratio}
% \vspace{-.75cm}
\end{figure}

\begin{figure}%[b]
  \centering
  %\vspace{10pt}
    \begin{minipage}{.45\textwidth}
      \centering
      %\vspace{10pt}
      %\hspace{-10pt}
      \vspace{-4pt}
      \includegraphics[width=\textwidth,keepaspectratio]{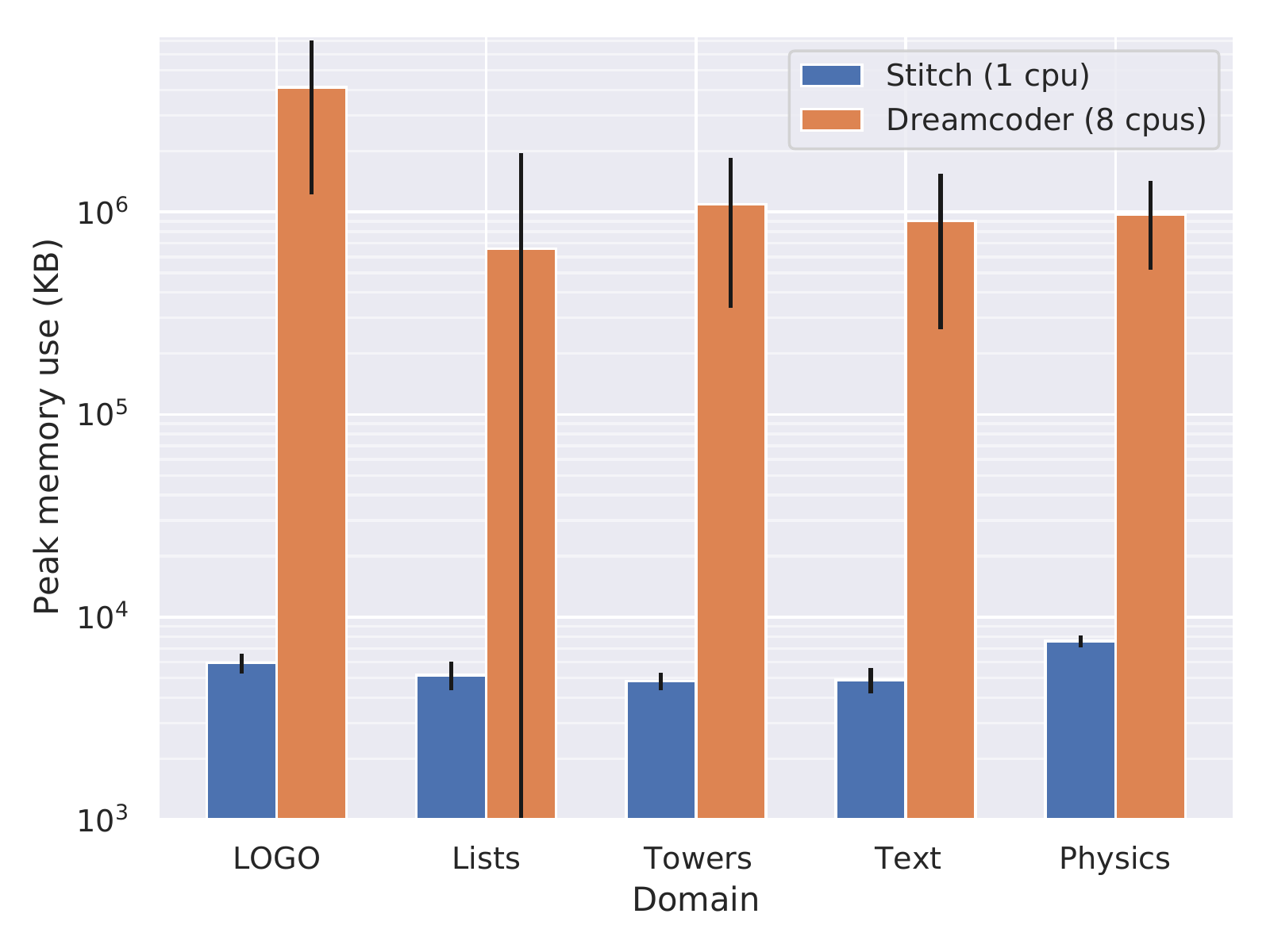}
      \caption{Peak memory usage of \stitch\ and DreamCoder while running on the five domains considered in \ref{subsec:claim-1}, averaged over all benchmarks. Lower is better; black lines indicate $\pm$ one standard deviation. Note the logarithmic y-axis. }
       \label{fig:casestudy1-memory}
     \end{minipage}
     \hspace{.5cm}
    \begin{minipage}{.45\textwidth}
      \centering
      \includegraphics[width=\textwidth,keepaspectratio]{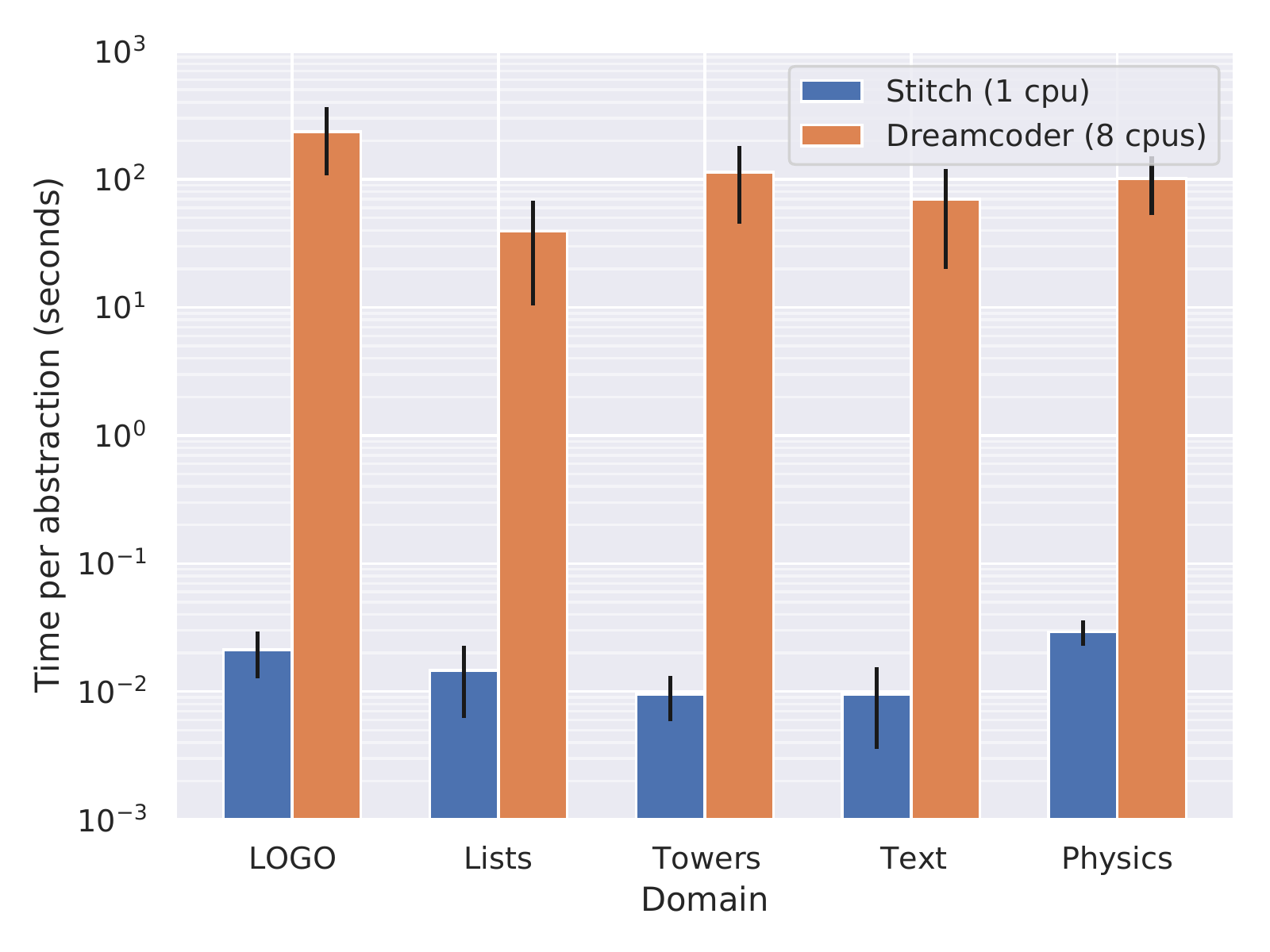}
      \caption{Wall-clock time required to find (and rewrite under) one abstraction in each of the five domains from \ref{subsec:claim-1}, averaged over all benchmarks. Lower is better; black lines indicate $\pm$ one standard deviation. Note the logarithmic y-axis.}
       \label{fig:casestudy1-time}
     \end{minipage}
    \vspace{-.25cm}
\end{figure}

\subsubsection*{Results}\label{subsec:casestudy1results} We compare DreamCoder  and \stitch{} for library quality and resource efficiency. 

\textit{\textbf{Library quality.}} We first examine the quality of the libraries learned by both \stitch\ and DreamCoder using the compression metric in \cref{eqn:min-metric}: \cref{fig:casestudy1-compression_ratio} shows the ratio between them across all of the benchmarks for each domain. A ratio of 1.0 indicates that programs are the \textit{exact same length} under the libraries learned by DreamCoder and \stitch, a ratio greater than 1.0 indicates that \stitch\ learns more compressive libraries, and a ratio less than 1.0 indicates that \stitch\ learns libraries which are less compressive. For example, a ratio of 1.1 indicates that DreamCoder rewrote to produce a corpus 10\% larger than that of \stitch, so it achieved less compression.

These results show that \stitch\ generally learns libraries of comparable and often greater quality than DreamCoder when matching the number of abstractions learned by the latter. In the \textit{logo}, \textit{towers}, \textit{text}, and \textit{physics} domains, \stitch\ always finds abstractions that are of equal---and often considerably greater---compressive quality than DreamCoder does; in the \textit{list} domain, \stitch\ more often than not still obtains better compression, but sometimes loses out to DreamCoder. This is a result of the fact that \stitch\ cannot learn higher-order abstractions, which are useful in this domain; although we emphasize our focus is on scalability, we will later present an extension of \stitch\ capable of handling most higher-order functions in \cref{subsec:claim-5}, based on the formalism developed in \cref{subsec:hybridapproach}. Nonetheless, we conclude that \textbf{\stitch\ learns libraries whose quality is comparable to and often better than those found by DreamCoder.}

\textit{\textbf{Resource efficiency.}} In addition to the quality of the libraries found, we are interested in how the two methods compare in terms of time and space requirements. Since we a priori believe \stitch\ to be significantly faster, for this evaluation we allow DreamCoder to use 8 CPUs but limit \stitch\ to single threading\footnote{While this may seem unfair to \stitch, it is worth noting that it would be unlikely to benefit from multithreading when running on the order of milliseconds anyway; DreamCoder, on the other hand, would struggle greatly in these domains without the aid of parallelism.}. The results are shown in \cref{fig:casestudy1-memory} and \cref{fig:casestudy1-time}; in summary, we find that \stitch\ takes \textit{tens of milliseconds} to discover abstractions across all five domains---achieving a \textit{3-4 order of magnitude} speed-up over DreamCoder---while also requiring more than \textit{2 orders of magnitude} less memory.
We thus conclude that \textbf{\stitch\ is dramatically more efficient than the state-of-the-art deductive baseline} when replicating the iterative library learning experiments of \citealt{ellis2021dreamcoder}.

\subsection{Large-Scale Corpus Library Learning}
\label{subsec:claim-2}

\subsubsection*{Experimental setup.}  While the previous experiment allowed us to benchmark \stitch\ directly against a state-of-the-art deductive baseline, the iterated learning setting considered by \citealt{ellis2021dreamcoder} only evaluates library learning on relatively small corpora of short programs discovered by the synthesizer.
Our second experiment is instead designed to evaluate \stitch\ in a more traditional learning setting, in which we aim to learn libraries of abstractions from a large corpus of existing programs all at once.

We source our larger-scale program datasets from \citealt{cogsci}, which present a series of datasets designed as a benchmark for comparing human-level abstraction learning and graphics program writing against automated synthesis and library learning models. These datasets are divided into two distinct high-level domains (\textit{technical drawings} and \textit{block-tower planning tasks}), each consisting of four distinct subdomains containing 250 programs:
\begin{itemize}
    \item \textbf{Technical drawing domains: \textit{nuts and bolts; vehicles; gadgets; furniture}}: CAD-like graphics programs that render technical drawings of common objects, written in an initial DSL consisting of looped transformations (scaling, translation, rotation) over simple geometric curves (lines and arcs). 
    \item \textbf{Tower construction domains: \textit{bridges; cities; houses; castles}}: Planning programs that construct complex architectures by placing blocks, written in an initial DSL that moves a virtual hand over a canvas and places horizontal and vertical bricks.
\end{itemize}
We choose these datasets not only for their size and scale (full dataset statistics in Table \ref{tab:domains}), but also for the complexity of their potential abstractions: \citealt{cogsci} explicitly design their corpora to contain complex hierarchical structures throughout the programs, making them an interesting setting for library learning.

\begin{table}[ht]
%\begin{small}
\begin{tabular}{@{}lrrrr@{}}
\toprule
 Domain       &  \#Programs &  Average program length & Average program depth  \\ \midrule
\rowcolor[HTML]{EFEFEF} 
nuts \& bolts & 250 & $76.03 \pm 24.22$ & $15.18 \pm 2.13$ \\
gadgets       & 250 & $142.85 \pm 87.32$ & $20.88 \pm 2.37$ \\
\rowcolor[HTML]{EFEFEF}
furniture     & 250 & $171.74 \pm 48.41$ & $31.83 \pm 5.33$ \\
vehicles      & 250 & $141.70 \pm 40.23$ & $21.22 \pm 1.35$ \\
\rowcolor[HTML]{EFEFEF}
bridges       & 250 & $137.03 \pm 59.71$ & $92.35 \pm 39.80$ \\
cities        & 250 & $161.70 \pm 55.56$ & $109.80 \pm 37.66$ \\
\rowcolor[HTML]{EFEFEF}
castles       & 250 & $189.09 \pm 60.18$ & $128.27 \pm 40.77$ \\
houses        & 250 & $168.13 \pm 55.75$ & $114.85 \pm 37.60$
\end{tabular}
%\end{small}
\caption{Summary statistics about the domains from \citealt{cogsci}. Program length is the number of terminal symbols in the program; program depth is the length of the longest path from root to leaf in the program tree. Both are reported as the mean over the entire dataset $\pm$ one standard deviation.}
\label{tab:domains}
\end{table}

When performing library learning in a synthesis setting by compressing a corpus of solutions, it's desirable to find abstractions that would be useful for solving \textit{new} tasks, as opposed to abstractions that overfit to the existing solutions. To evaluate how well the abstractions we learn apply to heldout programs in the domain, we split the corpora into train and test sets, running \stitch\ on the train set and evaluating its compression on the test set.
Since \citealt{cogsci} do not present a train/test split of their datasets, we use stochastic cross validation to evaluate the generalization of the libraries found by \stitch. For each domain, we randomly sample 80\% of the dataset to train on and reserve the last 20\% as a held out test set; we repeat this procedure 50 times. We then ask \stitch\ to learn a library consisting of $\leq 10$ abstractions with a maximum arity of 3, and average the results across the different random seeds.

To obtain a baseline to compare its performance against, we once again turned to DreamCoder \cite{ellis2021dreamcoder}. However, we found that DreamCoder was unable to discover even a single abstraction when run directly on any of the datasets from \citealt{cogsci}, despite being given hours of runtime and 256GB of RAM. We also experimented with heavily sub-sampling the training dataset before passing it to DreamCoder, but failed to find a configuration under which DreamCoder finds any interesting abstractions at all due to the fact that it immediately blows up on programs as long as these. As a result, we resort to presenting \stitch's performance metrics without any baseline to compare against; we stress that this is a direct result of the fact that \stitch\ is the first library learning tool that scales to such a challenging setting.

\subsubsection*{Results}

\begin{table}[ht]
%\begin{small}
\begin{tabular}{@{}lrrrr@{}}
\toprule
 \multirow{2}{*}{Domain}       & \multicolumn{2}{c}{Compression Ratio} & \multirow{2}{*}{Runtime (s)} & \multirow{2}{*}{Peak mem. usage (MB)}  \\ \cline{2-3}
              & Training set & Test set               &             & \\ \hline
\rowcolor[HTML]{EFEFEF}
 nuts \& bolts  & $   12.00 \pm 0.25  $  & $ 11.57 \pm 0.49$  & $  0.24 \pm 0.04$  & $    11.10 \pm 0.17   $  \\
   gadgets     & $    4.03 \pm 0.15  $  & $  3.91 \pm 0.31$  & $  1.83 \pm 0.35$  & $    30.11 \pm 0.85    $ \\
\rowcolor[HTML]{EFEFEF}
  furniture    & $    4.95 \pm 0.07  $  & $  4.85 \pm 0.26$  & $  2.83 \pm 0.49$  & $    30.88 \pm 0.62    $ \\
   vehicles    & $    4.28 \pm 0.14  $  & $  4.14 \pm 0.25$  & $  1.52 \pm 0.29$  & $    26.98 \pm 0.88    $ \\
\rowcolor[HTML]{EFEFEF}
   bridges     & $    4.36 \pm 0.06  $  & $  3.78 \pm 0.14$  & $ 17.58 \pm 1.26$  & $   189.08 \pm 10.10   $ \\
    cities     & $    3.15 \pm 0.05  $  & $  3.06 \pm 0.14$  & $ 50.74 \pm 3.96$  & $    413.63 \pm 9.51   $ \\
\rowcolor[HTML]{EFEFEF}
   castles     & $    2.57 \pm 0.07  $  & $  2.55 \pm 0.08$  & $ 77.26 \pm 6.98$  & $   683.89 \pm 32.96   $ \\
    houses     & $    8.92 \pm 0.21  $  & $  8.85 \pm 0.57$  & $ 15.54 \pm 1.37$  & $    241.17 \pm 2.86   $ \\

\end{tabular}
%\end{small}
\vspace{0.5em}
\caption{Results for the large-scale library learning experiment in Sec. \ref{subsec:claim-2}. The compression ratio refers to how many times smaller the corpus is after rewriting under the learned library compared to the original corpus; higher is thus better. All results are given as the mean $\pm$ one standard deviation over 50 runs with different random seeds for the dataset splitting.}
\label{tab:claim-2}
\end{table}

The results are summarized in Table \ref{tab:claim-2}. We find that \stitch\  scales up to even the most complex sub-domains, running in 77 seconds with a peak memory usage below 1GB on \textit{castles}. On four out of eight of the domains, \stitch\ finishes in single-digit seconds and consumes only tens of megabytes. This stands in stark contrast to DreamCoder, which we were unable to run on the very simple nuts \& bolts domain even with 256GB RAM and several hours worth of compute budget. These results thus support our claim that \textbf{\stitch\ scales to corpora of programs that would be intractable with prior library learning approaches}.

We hope that by providing our results on these datasets in full, future work in this field will benefit from having a directly comparable baseline.

\subsection{Robustness to Early Search Termination}
\label{subsec:claim-3}

\subsubsection*{Experimental setup}
This experiment is designed to evaluate how early into the search procedure \stitch\ finds what will eventually prove to be the optimal abstraction.
This is highly relevant in settings where the set of training programs is too large to run the search to completion. The experiment showcases one of \stitch's more subtle strengths: corpus-guided top-down abstraction search is an \emph{anytime} algorithm, and thus does not need to be run to completion to give useful results.

We re-use the domains from \citealt{cogsci} and once again evaluate the quality of the library learned (measured in program compression), similarly to what was done in the previous experiment. However, since we are interested in how quickly \stitch\ finds \emph{what it perceives to be} the optimal abstraction, we measure compressivity of the training dataset itself (rather than a held-out test set) and capture the compression ratio obtained by each \emph{candidate} abstraction found \textit{during} search (rather than just the compression ratio obtained when search has been run to completion). Thus, we are able to investigate how early on during the search procedure \stitch\ converges on a chosen library. We restrict \stitch\ to learning a single abstraction with a maximum arity of 3.

\subsubsection*{Results}

\begin{figure}[t]
    \centering
    \includegraphics[width=0.6\textwidth]{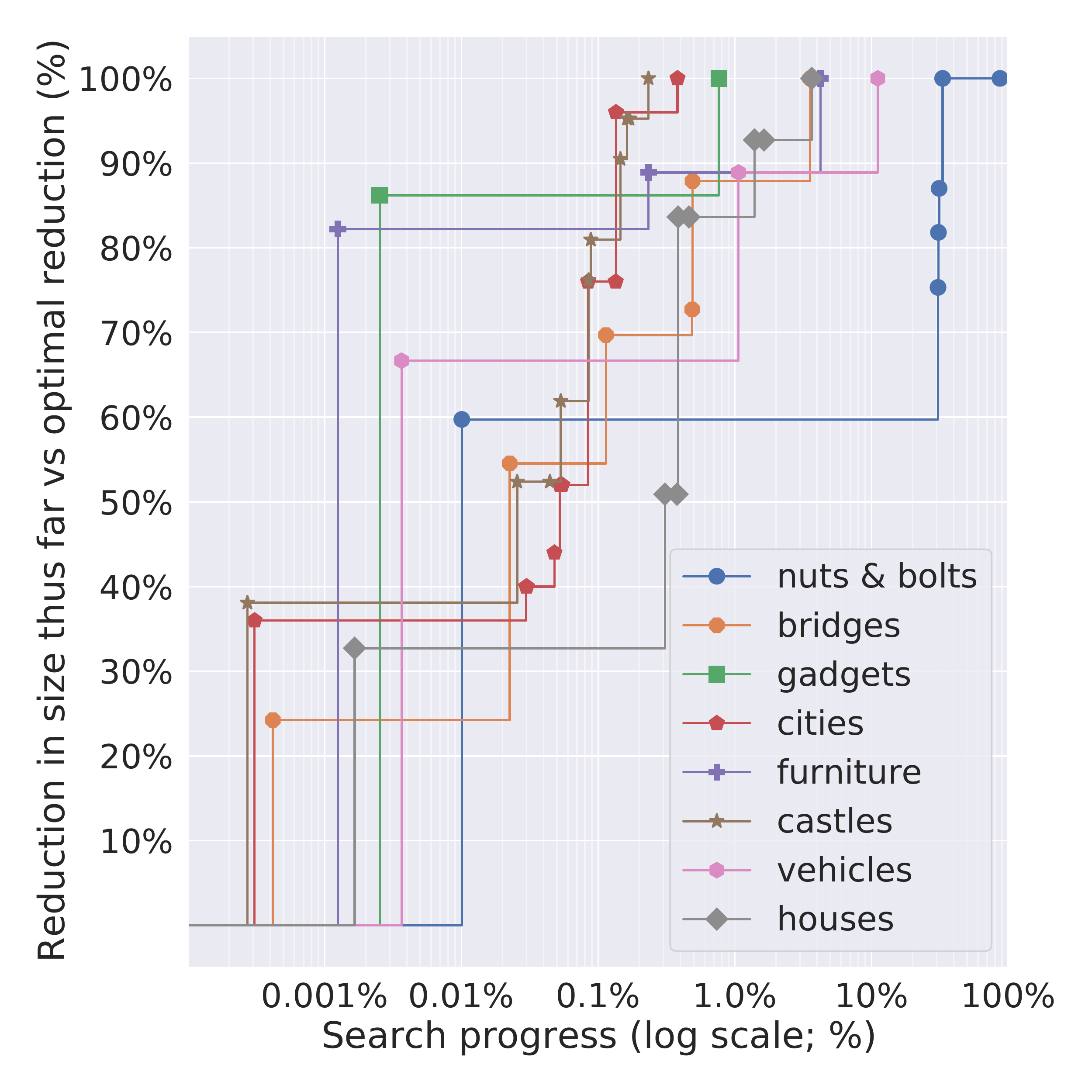}
    \caption{Reduction in size obtained on the training set when rewritten under the best abstraction found thus far vs. the number of nodes expanded during search. The y-axis is normalized with respect to the optimal abstraction; the x-axis is normalized with respect to the total number of nodes explored by \stitch. Lines thus end earlier the quicker \stitch\ finds the optimal abstraction; however in all runs \stitch\ \textit{continues to search} until reaching 100\% on the x-axis, exploring the rest of the abstraction space without finding a new best abstraction. Note the logarithmic x-axis.}
    \label{fig:cs2-ex3}
% \vspace{-.75cm
\end{figure}

The results are shown in Fig. \ref{fig:cs2-ex3}. These results validate our hypothesis that \textbf{\stitch\ is empirically robust to terminating the search procedure early}: in every sub-domain except for \emph{nuts \& bolts} \stitch\ converges to the optimal abstraction very early on, having only completed a tiny fraction of the total search.\footnote{Given that significant attention has already been given to wall-clock run-times of \stitch\ on similar workloads in \ref{subsec:claim-2}, we here use the number of nodes explored instead of wall-clock time to ensure deterministic and easily reproducible results.} We believe that this has great importance for \stitch 's applicability in data-rich settings since it suggests that a nearly-optimal abstraction can be found even if the search must be terminated early (e.g. after a fixed amount of time has passed), making early stopping an empirically useful way of speeding up the library learning process.

\subsection{Ablation Study}
\label{subsec:claim-4}

\subsubsection*{Experimental setup}
\textsc{\stitch} implements several different optimizations, which we have argued hasten the search for abstractions.
To verify that this claim holds in practice, we now carry out a brief ablation study.

Since the space of every possible combination of optimizations is too large to present succinctly, we focus our attention on four ablations:
\begin{itemize}
    \item \textbf{no-arg-capture} (from \cref{subsec:strict-dom-compression}), which disables the pruning of abstractions which are only ever used with the exact same set of arguments (and these arguments could therefore just be in-lined for greater compression).
    \item \textbf{no-upper-bound}, which disables the upper bound based pruning.
    \item \textbf{no-redundant-args} (from \cref{subsec:strict-dom-compression}), which disables the pruning of multi-argument abstractions that have a redundant argument that could be removed because it is always the same as another argument.
    \item \textbf{no-opts}, which disables all of \stitch's optimizations.
\end{itemize}
To isolate the impact of disabling optimizations, we run a single iteration of abstraction learning on each of the 8 domains from \ref{subsec:claim-2} and collect the number of nodes explored during the search. The first iteration is generally the most challenging as the corpus is large and has not yet been compressed at all. Focusing on the number of nodes explored (rather than for example runtime) allows for deterministic results. We also fix the maximum arity of abstractions to 3 in all runs, aiming to strike a good balance between compute requirements and how much the optimization will be exposed. We limit each run to 50GB of virtual address space, as well as 90 minutes of compute.

\begin{table}[tbp]
    \centering
    \begin{tabular}{l | c c c c c c c c }
     \diagbox{\small{System}}{\small{Domain}} 
     & \small{bridges} & \small{castles} & \small{cities}  & \small{gadgets}    & \small{furniture} & \small{houses}   & \makecell{\small{nuts \&}\\\small{bolts}} & \small{vehicles} \\ \hline
     \rowcolor[HTML]{EFEFEF}
     %baseline                 &  1.00  & 1.00   & 1.00  & 1.00     & 1.00      & 1.00    & 1.00       & 1.00 \\ 
     {\small no-arg-capture}              &  {\ram}   & {\ram}   & {\ram}   & 16714.28  & {\ram}       & {\ram}     & 30.21      & 26.67   \\
     %\rowcolor[HTML]{EFEFEF}
     {\small no-upper-bound}   &  12.15 &  23.71 & 27.90 &  207.76  &  119.00  &  36.18 &   179.12    & 179.99 \\
     \rowcolor[HTML]{EFEFEF}
     {\small no-redundant-args}  &  1.95  &  1.42  & 1.37  &   1.01   &   1.00   & 1.01  &    1.09     &  1.00   \\
     %\rowcolor[HTML]{EFEFEF}
     {\small no-opt}                            &  {\ram}   &  {\ram}  & {\ram}  & {\clk}     & {\ram}      & {\ram}    & {\clk}       & {\ram}   \\
     %no-opt                            &  DNF   &  DNF  & DNF  & DNF     & DNF      & DNF    & DNF       & DNF   \\
\end{tabular}
\vspace{0.5em}
    \caption{Results from the ablation study. Each cell contains the ratio between the number of nodes explored during search by that particular system and the number of nodes explored by the baseline on the same domain. Lower is better; 1.00 means performance is identical to the baseline. Cells labeled \ram \ crashed due to reaching the 50GB virtual address space limit, while those labeled \ \clksmol\ \ reached the 90 minute time restriction.}
    \label{tab:ablation}
\end{table}

\subsubsection*{Results}

The results are shown in \autoref{tab:ablation}. We first note as a sanity check that each ablation does indeed lead to reduced performance in general (i.e. explores a search space larger than the baseline does). Furthermore, the results suggest that upper bound based pruning is the most important in the \textit{nuts \& bolts} and \textit{vehicles} domains, while pruning out argument capture abstractions is the most important in the other six domains; it is noteworthy that this latter ablation by itself causes \stitch\ to reach the memory limit on more than half of the domains. 
On the other hand, disabling the pruning of redundant arguments has a relatively modest impact on the size of the search space, but still leads to an almost 2x improvement in the bridge domain.

Perhaps the most important takeaway from this ablation study is that when \textit{all} optimizations are disabled, \stitch\ fails to find an abstraction within the resource budget on any of the domains. This verifies our hypothesis that corpus-guided pruning of the search space is the key factor involved in making top-down synthesis of abstractions tractable.

\subsection{Learning Libraries of Higher-Order Functions}
\label{subsec:claim-5}

\subsubsection*{Experimental setup}
Our experiments up till now have focused on performant and scalable library learning.
This comes at the expense of some expressivity: deductive rewrite systems can, in principle, express broader spaces of refactorings.
For example, a rewrite based on inverting $\beta$-reduction allows inventing auxiliary $\lambda$-abstractions, which helps with learning higher-order functions: in~\citealt{ellis2021dreamcoder}, DreamCoder is shown to recover higher-order functions such as \texttt{map}, \texttt{fold}, \texttt{unfold}, \texttt{filter}, and \texttt{zip\_with}, starting from the Y-combinator.
This works by constructing a version space which encodes every refactoring that is equivalent up to $\beta$-inversion rewrites.
But DreamCoder's coverage comes at a steep cost as inverting $\beta$-reduction is expensive.

In this experiment, we seek to give evidence that it is possible to make \stitch\ recover all of these higher-order functions by layering it on top of the version space obtained after a single step of DreamCoder's $\beta$-inversion, following the formalism outlined in \cref{subsec:hybridapproach}.
We then compare this modified version of \stitch\ with DreamCoder on its ability to learn these higher-order functions from programs generated by intermediate DreamCoder iterations, and measure the runtime of each approach.

\subsubsection*{Probabilistic Re-ranking}
While the approach outlined in \cref{subsec:hybridapproach} should suffice to layer \stitch\ on top of deductive rewrite systems based on version spaces, some extra care needs to be taken to combine it with DreamCoder. This is because DreamCoder implements a probabilistic Bayesian objective for judging candidate abstractions (exploiting the connection between compression and probability~\cite{shannon1948mathematical}), seeking the library $L$ which maximizes $P(L)\prod_{p\in \corpus}P(p|L)$ for a given or learned prior $P(L)$ and program likelihood $P(p|L)$. \stitch, on the other hand, effectively judges compression quality via a cost function capturing the (weighted) size of the programs as detailed in \cref{sec:framework-instantiation}.

To implement this probabilistic heuristic in \stitch, we simply run \stitch\ on the version spaces as-is but then re-score each complete abstraction popped off of the priority queue under the Bayesian objective, using DreamCoder's models of the prior and likelihood. To make this integration easier, we re-implement \stitch\ in Python, giving a prototype version called \textsc{pyStitch} which is only used for this experiment. Our implementation employs strict dominance pruning, and prunes using the bound on the approximate utility given in \cref{subsec:hybridapproach}.

%use upper bound pruning and strict dominance pruning in our implementation, adapted to the version space data structure. %the exact utility calculation as our upper bound for pruning, and do not employ strict dominance pruning.

\begin{table}[ht]
\begin{tabular}{llcccccr}\toprule
\multicolumn{2}{c}{System} & Fold &Unfold &Map &Filter &ZipWith &Time (s) \\ \midrule 
\rowcolor[HTML]{EFEFEF} 
 &Base &\checkmark &$\times$ &\checkmark &$\times$ &\checkmark & 503 \\
 \rowcolor[HTML]{EFEFEF} 
&+Bayes &\checkmark &$\times$ &\checkmark &$\times$ &\checkmark &817 \\
\rowcolor[HTML]{EFEFEF} 
&+VS &\checkmark &\checkmark &\checkmark &$\times$ &\checkmark &3231 \\
\rowcolor[HTML]{EFEFEF} 
\multirow{-4}{*}{\textsc{pyStitch}} &+Bayes+VS &\checkmark &\checkmark &\checkmark &\checkmark &\checkmark  &3042 \\
\multirow{4}{*}{DreamCoder} &Step 1 &$\times$ &$\times$ &\checkmark &$\times$ &\checkmark & 67 \\
&Step 2 &$\times$ &$\times$ &\checkmark &\checkmark &$\times$ & 116 \\
&Step 3 &$\times$ &$\times$ &\checkmark &\checkmark &\checkmark & 2254 \\
&Step 4 &\checkmark &\checkmark &\checkmark &\checkmark &\checkmark & 228048 \\
\bottomrule
\end{tabular}
\caption{Comparing library learners on functional programming exercises. DreamCoder, $n$-steps: deductive baseline rewriting $n$ steps of $\beta$-reduction.
\textsc{pyStitch}: Python reimplementation of \stitch, which enables better interoperability with DreamCoder's version space algebra (+VS) and probabilistic models (+Bayes).
\textsc{pyStitch}+VS+Bayes learns all of the same higher-order functions, using $<2\%$ of the compute.
}\label{tab:pystitch}
\end{table}

\subsubsection*{Results}
The results are shown in \cref{tab:pystitch}. We note first the large discrepancy in runtimes; running DreamCoder with 4 steps of rewriting takes roughly 2.5 days of compute on this domain, while even the slowest version of \textsc{pyStitch} still finishes in less than an hour. In terms of the number of higher-order abstractions found, DreamCoder only finds all five when run in its most expensive configuration; reducing the computation cost quickly decreases its expressivity. For \textsc{pyStitch}, the Bayesian re-ranking alone (\textsc{pyStitch}+Bayes) does not yield any improvements, while running it on the version spaces (\textsc{pyStitch}+VS) yields 4 out of 5 functions. However, it is only when these two adaptations are used in conjunction (\textsc{pyStitch}+Bayes+VS) that the \stitch-based method is able to find all of the higher-order functions.

In summary, we find that running CTS after a single step of version space rewriting and then probabilistically re-ranking the results suffices to recover the core higher-order functions that DreamCoder learns, while using $<$2\% of its compute. We thus conclude that \textbf{running \stitch\ on top of a deductive rewrite system reduces the expressivity gap, while retaining superior performance}.

\section{Related Work}
\label{sec:related-work}
\stitch\ is related to two core ideas from prior work: \textbf{deductive refactoring and library learning systems}, which introduce the idea of learning abstractions that capture common structure across a set of programs, but have largely been driven by deductive algorithms; and \textbf{guided top-down program synthesis systems}, which use cost functions to guide top-down enumerative search over a space of programs, but have largely been used to synthesize whole programs for individual tasks in prior work.

\subsection{Deductive Refactoring and Library Learning}
Recent work shares \stitch's goal of learning \textit{libraries} of program abstractions which capture reusable structure across a corpus of programs \cite{ellis2018library,ellis2021dreamcoder,dechter2013bootstrap,cropper2019playgol,shin2019program,allamanis2014mining,iyer2019learning,wong2021leveraging, jones2021shapemod}. Several of these prior approaches also introduce a utility metric based on \textit{program compression} in order to determine the most useful candidate abstractions to retain \cite{dechter2013bootstrap,ellis2021dreamcoder,wong2021leveraging,lazaro2019beyond,iyer2019learning}.

Much of this this prior work follows a bottom-up approach to abstraction learning, combining a bottom-up traversal across individual training programs with a second stage to extract shared abstractions from across the training corpus. This approach includes systems that work through direct \textit{memoization} of subtrees across a corpus of programs \cite{dechter2013bootstrap,lin2014bias,lazaro2019beyond}; \textit{antiunification} (caching tree templates that can be unified with training program syntax trees) \cite{ellis2018library,henderson2013cumulative,hwang2011inducing,iyer2019learning}; or by more sophisticated \textit{refactoring} using one or more rewrite rules to expose additional shared structure across training programs \cite{ellis2018library,ellis2021dreamcoder, chlipala2017end, liang2010learning}.

Many of the bottom up algorithms draw more generally on \textit{deductive synthesis} approaches that apply local rewrite rules in a bottom-up fashion to program trees in order to refactor them --- historically, to synthesize programs from a declarative specification of desired function \cite{burstall1977transformation,manna1980deductive}. Deductive approaches to library learning, however, confront fundamental memory and search-time scaling challenges as the corpus size and depth of the training programs increases; prior deductive approaches such as \cite{ellis2018library,ellis2021dreamcoder} use version spaces \cite{lau2003programming, mitchell1977version} to mitigate the memory usage during bottom-up abstraction proposal. Still, deductive approaches are challenging to bound and prune (unlike the top-down approach we take in \stitch ), as they generally traverse individual program trees locally and must store possible abstraction candidates in memory before the extraction step. 

Some prior work \cite{shin2019program,allamanis2014mining} takes an MCMC-based approach with better scaling behavior than deductive rewrites, but differs from the goals of \stitch\ as they don't address binders and focus on common syntactic fragments instead of well-formed functions.

\subsection{Guided Top-Down Program Synthesis}
\stitch\ uses a corpus-guided top-down approach to learning library abstractions that is closely related to recent \textit{guided enumerative synthesis techniques}. This includes methods that leverage type-based constraints on holes \cite{feser2015synthesizing,polikarpova2016program}, over- and under-approximations of the behaviors of holes \cite{lee2016,chen2020}, and probabilistic techniques to heuristically guide the search \cite{balog2016deepcoder,ellis2020dreamcoder,ellis2021dreamcoder,nye2021blended,2020swaratneuraladmissableheuristics}. These approaches have largely been applied in to synthesize entire programs based on input/output examples or another form of specification, in contrast to the abstraction-learning goal in our work \cite{allamanis2018survey,balog2016deepcoder,chen2018execution,ellis2018library,ganin2018synthesizing,koukoutos2017repair}. Like \stitch\, however, these approaches sometimes use cost functions (such as the likelihood of a partially enumerated program under a hand-crafted or learned probabilistic generative model over programs) in order to direct search towards more desirable program trees. However, \stitch's cost function leverages a more direct relationship between partially-enumerated candidate functions and the existing training corpus, unlike the cost functions typically applied in inductive synthesis, which must be estimated from input/output examples.

\subsection{Lambda-Aware Unification}

The \lambdaU\ procedure presented in \cref{sec:technical} relates to prior work on unification modulo binders \cite{Huet75,Miller91,Miller92,Dowek95,Dowek96}.

\textit{Notion of beta-equivalence.} This prior work is concerned with more general notions of equivalence modulo beta-reduction, while \lambdaU\ is based on a restricted but fast syntax-driven equivalence. For example, in \citealt{Dowek96} one might try to unify $(\alpha\ \texttt{foo})$ with $(\texttt{foo})$ and get the two solutions: $[\alpha \to (\lambda.\ \$0)]$ and $[\alpha \to (\lambda.\ \texttt{foo})]$. In contrast, \lambdaU\ will never introduce a $\lambda-$abstraction to create a higher order argument and thus $(\alpha\ \texttt{foo})$ and $(\texttt{foo})$ would not unify at all, since one is an application while the other is a primitive. Instead, higher order abstraction arguments are handled by combining the algorithm with a deductive approach when desirable as in \cref{subsec:hybridapproach}. We also note that in $\lambdaU(A,e)$ it is assumed that $e$ is an expression and thus does not contain abstraction variables, which further simplifies the approach compared to this prior work.

\textit{Handling of binders.} \citealt{Dowek95} and \citealt{Dowek96} both employ de Bruijn indexed variables and therefore must similarly account for the shifting of variables in arguments when inverting beta reduction. While \lambdaU\ handles this through \textsc{DownshiftAll}, \citealt{Dowek95} and \citealt{Dowek96} both converting their terms to the $\lambda\sigma-$calculus of explicit substitutions \cite{abadi1989explicitsubst} which allows them to insert upshifting operators at each abstraction variable location. This alternate handling of shifting is useful given the more general notion of equivalence they are considering, but is excessive for our simpler syntax-guided task.

\textit{Handling of holes.} In $\lambdaU(A,e)$ we allow $A$ to contain holes $\hole$ which are allowed to violate index shifting rules during unification as they are considered unfinished subtrees of the abstraction. This is handled through the use of $\&i$ indices. While the special-case handling of holes in expressions is not directly part of this prior work, there are similarities between it and the \textit{Skolemization} \cite{skolem1920} done in \citealt{Miller92}. Skolemization allows for lifting an existential quantifier (i.e. an abstraction variable or hole) above a universal quantifier (i.e. a lambda) by turning the abstraction variable or hole into a function of its local context --- in essence piping the local context into the hole. In the context of the lambda calculus this is essentially a form of lambda lifting \cite{Johnsson1985}, which is the process of lifting a local function that contains free variables by binding the free variables as additional arguments to the function and passing them in at each call site. In \citealt{Miller92} this is used for abstraction variables (as there are no holes) and aids in the more general notion of beta-equivalence they are dealing with, while for our purposes the $\&i$ variables suffice and don't require extra manipulations of lambdas to thread in the local context.

\subsection{Upper Bounds in Network Motifs}
There is also an interesting connection between this work and prior work on finding \textit{network motifs} \cite{networkmotifs2002}, which are frequently-occurring subgraphs within a corpus of graphs. One approach to mining high-frequency motifs  is based on growing a motif one edge at a time, much like we grow abstractions one node at a time during our search \cite{networkmotifs2005,kuramochi2004,kuramochi2001}. This approach makes use of the insight that the frequency of a motif will only decrease as it is grown, giving an upper bound on the frequency of any motif derived from it. While we employ a more complex compressive utility function than their frequency-based utility, our upper bound is based on a similar insight that larger abstractions will match at a subset of the locations.

% There are a number of differences: their utility function is based only on the number of usages as opposed to compression. They aren't using the lambda calculus and don't handle binders. They're working over graphs instead of trees. They dont apply strict dominance pruning. While we can solve selecting rewrite locations optimally through dynamic programming, they face a similar problem version on unlabelled directed graphs is NP-complete and they must only approximate it. 

\subsection{Comparison to \textsc{babble}}\label{subsec:babble}

\textsc{babble} \cite{babble2023} is concurrent work in library learning that likewise adopts the compression objective from DreamCoder. While \textsc{babble} focuses primarily on expressive library learning through an algorithm that can reason over semantic equivalences represented through e-graphs, \stitch\ focuses primarily on efficient library learning through a parallel anytime branch-and-bound algorithm.

\stitch\ synthesizes the maximally-compressive abstraction for a corpus of programs through a branch-and-bound top-down search, computing an upper bound on the compression of any partially-constructed abstraction to guide and prune the search. In contrast, \textsc{babble} operates on an e-graph representing many equivalent corpora and uses anti-unification to propose candidate abstractions and a custom e-graph beam-search extraction algorithm to select a maximally-compressive set of abstractions from the proposal.

At a high level, \stitch\ aims to be \textit{efficient} while \textsc{babble} aims to be \textit{expressive}. The key differences between the algorithms are:
\begin{itemize}

    \item \textit{Efficiency}: On DreamCoder’s benchmarks, \textsc{babble} is 10-100x faster than DreamCoder while \stitch\  is 1,000-10,000x faster than DreamCoder. Unlike \textsc{babble}, \stitch\ is an anytime algorithm that can be stopped early for a best-so-far result.\vspace{5pt}
    
    \item \textit{Expressivity}: \textsc{babble} can learn libraries modulo equational theories, allowing it to find common semantic abstractions despite syntactic differences in the programs. \stitch\ only explores non-syntactic equivalences preliminarily in the \textsc{PyStitch} prototype with the DreamCoder beta-inversion rewrite rule (\cref{subsec:hybridapproach} and \cref{subsec:claim-5}).\vspace{5pt}
    
    \item \textit{Jointly learning a library}: \textsc{babble} can learn multiple abstractions at once that jointly provide compression, while DreamCoder and \stitch\  repeatedly learn a single abstraction at a time. \stitch\  provides an optimality guarantee on each abstraction learned while \textsc{babble} approximates the joint objective through a beam search.
\end{itemize}

These two approaches has advantages and disadvantages, and we believe that there is strong potential in combining the two.

\section{Conclusion and Future Work}\label{sec:conclusion}
We have presented \textit{corpus-guided top-down synthesis} (CTS)---an efficient new algorithm for synthesizing libraries of functional abstractions capturing common functionality within a corpus of programs. CTS directly synthesizes the abstractions, rather than exposing them through a series of rewrites as is done by deductive systems. Key to its performance is the usage of a guiding utility function, which allows CTS to effectively search (and prune out large portions of) the space of possible abstractions.

% \stitch\ is open-source and available both as a Python package and a Rust crate; the code, installation instructions, and tutorials are available at the GitHub\footnote{\url{https://github.com/mlb2251/stitch}}

We implement this algorithm in \stitch, an open-source library learning tool with an associated Python library and Rust crate.\footnote{\url{https://github.com/mlb2251/stitch}} We evaluate \stitch\ across five experimental settings, demonstrating that it learns comparably compressive libraries with 2 orders of magnitude less memory and 3-4 orders of magnitude less time compared to the state-of-the-art deductive algorithm of \citealt{ellis2021dreamcoder}. We also find that \stitch\ scales to learning libraries of abstractions from much larger datasets of deeper program trees than is possible with prior work, and that the \textit{anytime} property of corpus-guided top-down search --- abstractions discovered via top-down search are already compressive early in search and improve as it continues ---  offers the opportunity for high-quality library learning even in complex domains through early search termination.

There remain many open problems and exciting directions in abstraction learning. One major direction is extending the CTS approach to handle reasoning over more general deductive rewrite systems, like those encoded in \texttt{egg} \cite{10.1145/3434304} as \textsc{babble} does \cite{babble2023} (also see \cref{subsec:babble}). The \textsc{PyStitch} prototype presented in \cref{subsec:hybridapproach} only works with a single rewrite rule (beta-inversion) and uses an approximate utility without an optimality guarantee. This is because it is challenging to adapt the upper bound used by \stitch\ to version spaces or e-graphs due to the trade-off between the compression gained from equational rewrites and compression gained from abstraction learning.

The speedup afforded by \stitch\ may allow for new applications of abstraction learning as well, as it can now be used as an inexpensive subroutine that can be called much more frequently than prior algorithms. To further extend its range of applications, CTS could also be adapted to finding abstractions in data structures beyond program trees, such as dataflow DAGs and more general graph structures like molecules, if the bounds and matching can be adapted. Finally, while we work with fairly large lambda calculus programs in this work, there is a clear gap in scale between these programs and those in an actual codebase, so exploring applications of CTS to real-world code is an exciting direction.

\section*{Data Availability Statement}

All source code can be found at \url{https://github.com/mlb2251/stitch}. An artifact for reproducing results in this work is available at \url{https://github.com/mlb2251/stitch-artifact} or alternatively a static version is available at \url{https://zenodo.org/record/7151663} \cite{stitchartifact}.

%% Acknowledgments
\begin{acks}                            %% acks environment is optional
                                        %% contents suppressed with 'anonymous'
  %% Commands \grantsponsor{<sponsorID>}{<name>}{<url>} and
  %% \grantnum[<url>]{<sponsorID>}{<number>} should be used to
  %% acknowledge financial support and will be used by metadata
  %% extraction tools.

We thank A. Lew and J. Andreas for helpful discussions, and J. Feser and I. Kuraj for feedback on the manuscript. M.B. and G.G are supported by the \grantsponsor{NSF}{National Science Foundation}{} (NSF) Graduate Research Fellowship under Grant No. \grantnum{NSF}{2141064}. M.B. is also supported by the \grantsponsor{DARPA}{Defense Advanced Research Projects Agency}{} (DARPA) under the SDCPS Contract \grantnum{DARPA}{FA8750-20-C-0542}. T.X.O. is supported by Herbert E. Grier (1933) and Dorothy J. Grier Fund Fellowship. L.W. and J.B.T. are supported by \grantsponsor{AFOSR}{AFOSR}{} under grant number \grantnum{AFOSR}{FA9550-19-1-0269}, the MIT Quest for Intelligence, the MIT-IBM Watson AI Lab, ONR Science of AI, and DARPA Machine Common Sense. G.G. is also supported by the MIT Presidential Fellowship. A.S. is supported by the NSF under Grant No. \grantnum{NSF}{1918839}. Any opinions, findings, and conclusions or recommendations expressed in this material are those of the author(s) and do not necessarily reflect the views of sponsors.

\end{acks}

\pagebreak

%% Bibliography
\bibliography{bib}

\pagebreak
%% Appendix
\appendix
\section{Corpus-guided top-down search: full algorithm}
\label{appendix:cts-full}

In this appendix, we give an expansive description of the full corpus-guided top-down search (CTS) algorithm to ease re-implementation.

The full algorithm is given in \cref{alg:corpus_guided}. It takes as input a corpus of programs $\corpus$, a utility function $\util{A}$, and utility upper bound function $\upperbound{\abshole}$. The algorithm searches for and returns the complete abstraction $A_{best}$ that maximizes the utility function. The algorithm maintains a priority queue $Q$ of partial abstractions ordered by their utility upper bound (or alternatively a custom priority function), a best complete abstraction so-far $A_{best}$, and a corresponding best utility so far $U_{best}$ for that abstraction. The priority queue is initialized to hold the single-hole partial abstraction $\hole$, from which any abstraction can be derived. $A_{best}$ and $U_{best}$ are initialized as the best arity-zero abstraction and corresponding utility (line \ref{alg:corpus_guided:arity-zero}), since arity-zero abstractions can be quickly and completely enumerated (each unique, closed subexpression in the corpus is an arity-zero abstraction).

We then proceed to the core loop of the algorithm from lines \ref{alg:corpus_guided:while_start}-\ref{alg:corpus_guided:while_end}. At each step of this loop we pop the highest-priority partial abstraction $\abshole$ off of the priority queue and process it. We discard $\abshole$ if its utility upper bound doesn't exceed our best utility found so far (lines \ref{alg:corpus_guided:bound1_start}-\ref{alg:corpus_guided:bound1_end}). We then choose a hole $h$ in $\abshole$ to expand with the procedure \textsc{Choose-Hole}. \textsc{Choose-Hole} can be a custom function; we find from preliminary experiments on a subset of the datasets that choosing the most recently introduced hole is effective in practice.

The algorithm then uses the procedure \textsc{Expansions} to iterate over all possible single step expansions of the hole $h$ in $\abshole$, such as replacing the hole with \texttt{+} or with (\texttt{app} \hole\ \hole). For each expanded abstraction $\abshole'$, it is easy to compute its set of match locations $\matches{\abshole'}$ since we know this will be a subset of the match locations of $\abshole$ (and these will be \textit{disjoint} subsets, except when expanding into an abstraction variable $\alpha$). We can easily inspect the relevant subtree at each match location of the original abstraction to see which expansions are valid and which match locations will be preserved by a given expansion. Note that an expansion to a free variable can be pruned immediately, as any resulting abstraction will not be well-formed.

When expanding to an abstraction variable $\alpha$, if $\alpha$ is a new variable not already present in the partial abstraction, the set of match locations is unchanged. If $\alpha$ is an existing abstraction variable then this is a situation where the same variable is being used in more than one place, as in the \textit{square} abstraction $(\lambda \alpha.\ \texttt{*}\ \alpha\  \alpha)$. In this case we restrict the match locations to the subset of locations where within a location all instances of $\alpha$ match against the same subtree. Additionally, if a maximum arity is provided then an expansion that causes the abstraction to exceed this limit is not considered.

When considering each possible expanded abstraction $\abshole'$, there is room for strong corpus-guidance. First, we don't need to consider any expansions that would result in zero match locations since all abstractions in this branch of search will have zero rewrite locations per \cref{lemma:matches-subset-rewrites} (lines \ref{alg:corpus_guided:no_match_start}-\ref{alg:corpus_guided:no_match_end}). Furthermore, we use $\upperbound{\abshole'}$ to upper bound the utility achievable in this branch of search and discard it if it is less than our best utility so far (lines \ref{alg:corpus_guided:bound2_start}-\ref{alg:corpus_guided:bound2_end}). Since each $\abshole'$ covers a disjoint set of match locations (except in the case of expanding into an abstraction variable), the set of match locations often drops rapidly and can allow for the calculation of a tight upper bound depending on the utility function. As a final step of pruning, if we can identify that $\abshole'$ is \textit{strictly dominated} by some other abstraction $\abshole''$, we may discard $\abshole'$ (lines \ref{alg:corpus_guided:inferior_start}-\ref{alg:corpus_guided:inferior_end}). Note additionally that arity-zero abstractions were precomputed, then abstractions that match at a single location are safe to prune as well (as long as they don't have free variables) as arity-zero abstractions are always superior for single match locations.

For any partial abstraction $\abshole'$ that has not been pruned, we then check whether it is a complete abstraction (there are no remaining holes) or whether it is still a partial abstraction. Partial abstractions get pushed to the priority queue ranked by their utility upper bound, and complete abstractions are used to update $A_{best}$ and $U_{best}$ if they have a higher utility than any prior complete abstractions.

Once there are no more partial abstractions remaining in the priority queue, the algorithm terminates. Note that since the algorithm maintains a best abstraction so-far, it can also be terminated early, making it an \textit{any-time algorithm}. We also note that this algorithm is amenable to parallelization as it can be easily parallelized over the while loop on lines \ref{alg:corpus_guided:while_start}-\ref{alg:corpus_guided:while_end}, especially since the algorithm remains sound even when the upper bound $U_{best}$ used pruning is not always up to date.

This algorithm can be iterated to build up a library of abstractions.

Finally note that in our implementation, we employ structural hashing to simplify equality checks between subtrees and avoid re-doing work on multiple identical copies of a subtree. Our complete implementation of CTS for compression is provided at \url{https://github.com/mlb2251/stitch}.

\begin{algorithm}
\caption{Corpus-guided top-down abstraction synthesis. Color-coded: \textcolor{StitchDarkBlue}{Upper bound pruning}, \textcolor{StitchLightBlue}{Zero-usage pruning}, \textcolor{StitchPurple}{Strict dominance pruning}}\label{alg:corpus_guided}
\begin{algorithmic}[1]
\Require Corpus of input programs $\corpus$, utility function $\util{A}$, and utility upper bound function $\upperbound{\abshole}$
\Ensure The maximally compressive abstraction $A_{best}$ 
% \State \textbf{procedure} \textsc{Top-Down-Search}
\State $Q \gets$ Priority-Queue \{ \hole\ \} \Comment{New priority queue with the single partial abstraction \hole} 
\State $A_{best} \gets  \textsc{Best-Arity-Zero-Abstraction}(\corpus)$ \Comment{ Initialize best abstraction so far} \label{alg:corpus_guided:arity-zero}
\State $U_{best} \gets \util{A_{best}}$
\While{$\text{non-empty}(Q)$} \label{alg:corpus_guided:while_start}
    \State $\abshole \gets \text{pop-max}(Q) $  \Comment{ Next partial abstraction to expand }
    \color{StitchDarkBlue}
    \If{ $\upperbound{\abshole} \leq U_{best}$} \label{alg:corpus_guided:bound1_start}
        \State \textbf{continue} \Comment{ Upper bound pruning }
    \EndIf \label{alg:corpus_guided:bound1_end}
    \color{black}
    
    \State $h \gets \textsc{Choose-Hole}(\abshole)$ \Comment{ Choose a hole to expand }
    \For{$(\abshole',M') \in \textsc{Expansions}(\abshole, h, \corpus)$}  \Comment{ abstraction $\abshole'$ and match locations $M'$ }
        \color{StitchLightBlue}
        \If{ $\text{length}(M') == 0$} \label{alg:corpus_guided:no_match_start}
            \State \textbf{continue} \Comment{ No match locations in corpus }
        \EndIf \label{alg:corpus_guided:no_match_end}
        \color{black}
        \color{StitchDarkBlue}
        \If{ $\upperbound{\abshole'} \leq U_{best}$}  \label{alg:corpus_guided:bound2_start}
            \State \textbf{continue} \Comment{ Upper bound pruning }
        \EndIf \label{alg:corpus_guided:bound2_end}
        \color{black}
        \color{StitchPurple}
        \If{ $\text{strictly-dominated}(\abshole',\corpus)$} \label{alg:corpus_guided:inferior_start}
            \State \textbf{continue} \Comment{ Strict dominance pruning }
        \EndIf \label{alg:corpus_guided:inferior_end}
        \color{black}
        \If{ $\text{has-holes}(A_e)$}
            \State $Q \gets  Q \cup \abshole' $ \Comment{ add partial abstraction to heap }
        \Else
            \color{cyan}
            % \State (Higher-order stuff would get handled right here if we add it; ignore this for now)
            \color{black}
            \If{ $\util{\abshole'} > U_{best} $}
                \State $U_{best} \gets \util{\abshole'} $ \Comment{ new best complete abstraction }
                \State $A_{best} \gets \abshole'$
            \EndIf
        \EndIf
        
    \EndFor
\EndWhile \label{alg:corpus_guided:while_end}
\end{algorithmic}

\end{algorithm}

\section{proof of Correctness of \lambdaU}
\label{appendix:lambdaunify-proof}

As introduced in \cref{sec:technical}, $\lambdaU(A,e)$ returns a mapping from abstraction variables and holes to expressions $[\alpha_i \to e_i',\ \ldots,\  \hole_j \to u_j',\ \ldots]$ such that 
\begin{equation}\label{eqn:beta-red-2}
    (\lambda \alpha_i.\ \ldots\ \lambda \hole_j.\ \ldots A)\ e_i'\ \ldots\ u_j'\ \ldots = e
\end{equation}
through beta reduction. The \lambdaU\ procedure is given in \cref{fig:lambda_unify} and the definition of beta reduction with a modified upshifting operator is given in \cref{fig:subst-upshift}.

We can write \cref{eqn:beta-red-2} in terms of substitution as

\begin{equation}
    [\alpha_i \to e_i',\ \ldots,\  \hole_j \to u_j',\ \ldots]\ \subst  A = e
\end{equation}

Which can in turn be written in terms of \lambdaU\ as

\begin{equation} \label{eqn:correctness}
    \lambdaU(A,e) \subst  A = e
\end{equation}

To prove the correctness of \lambdaU, we must show that \cref{eqn:correctness} holds for all $A$ and $e$.  We do this in two ways: a written proof (below) of our correctness theorem (\cref{lemma:correctness}), and a machine-verifiable proof in Coq. The Coq proof is given in \texttt{stitch.v} in the supplementary material and was proven with \texttt{CoqIDE 8.15.2}. The meanings of the definitions and lemmas in the Coq proof are commented. The Coq proof differs from the written proof in places because encoding the proof in Coq required coming up with some representations that worked well for Coq, though these differences are limited. For example, the Coq proof requires fully formalizing the behavior of \textsc{Merge} and proofs related to it which is more simply explained in the written proof. When lemmas in the written proof below correspond directly to lemmas in the Coq proof, the corresponding Coq lemma is indicated in the written proof below.

\subsection{Well-formedness}

We formalize a notion of \textit{well-formedness} to capture the set of expressions that are accessible starting from an expression with no $\&i$ variables then applying any series of upshifts and downshifts. We define $\text{WF}_d(e)$, the well-formedness of an expression $e$ at a depth $d$, as follows:

\begin{align*}
    &\text{WF}_d\ \lambda.b = \text{WF}_{d+1}\  b \\ 
    &\text{WF}_d\ (f x) = \text{WF}_d\ f \land \text{WF}_d\ x \\
    &\text{WF}_d\ \$i = \top\\ 
    &\text{WF}_d\ \&i = \begin{cases}
        \top, & \text{if } i < d \\
        \bot, & \text{if } i \geq d \\
    \end{cases}\\
    &\text{WF}_d\  t = \top,\ \text{for $t \in \gsym$} \\ 
\end{align*}

Note that trivially any expression $e$ without any $\&i$ variables is well formed, as $\&i$ variables are the only source of $\bot$ in this definition.

If every expression in a mapping $l$ is well-formed, we say that the mapping itself is well-formed, written $(\text{WFMap}\ l)$.

\subsection{Relevant Lemmas}

\ \\\noindent \textbf{The following two lemmas show that upshifting and downshifting preserve well-formedness.} In Coq these lemmas are \texttt{upshift\_wf} and \texttt{downshift\_wf}.

\begin{lemma}\label{lemma:up-well}
$\forall e.\ \forall d.\ \text{WF}_d\ e \implies  \text{WF}_d\ \uparrow_d e$
\end{lemma}
\begin{proof}
We proceed by induction on e.\\

\textbf{\textit{Case}} $e = \lambda.\ b$. Our goal is

$$\text{WF}_d\ \uparrow_d \lambda. b$$

We can move $\uparrow_d$ and $\text{WF}_d$ into the lambda by unfolding their definitions to get

$$\text{WF}_{d+1}\ \uparrow_{d+1} b$$

We also have the premise $\text{WF}_d\ \lambda. b$, where we can likewise move the $\text{WF}_d$ into the lambda to get $\text{WF}_{d+1}\ b$. The inductive hypothesis $\forall d.\ \text{WF}_d\ b \implies  \text{WF}_d\ \uparrow_d b$ can be instantiated with $d$ as $d+1$ to get $\text{WF}_{d+1}\ b \implies  \text{WF}_{d+1}\ \uparrow_{d+1} b$, the conclusion of which is precisely our goal and the premise of which is precisely our other premise.

\textbf{\textit{Case}} $e = f\ x$. Our goal is

$$\text{WF}_d\ \uparrow_d (f\ x)$$

We can move $\uparrow_d$ and $\text{WF}_d$ into the application by unfolding their definitions to get

$$(\text{WF}_d\ \uparrow_d f)\ \land (\text{WF}_d\ \uparrow_d x)$$

We can use our inductive hypotheses $\text{WF}_d\ f \implies \text{WF}_d\ \uparrow_d f$ to prove the left side of this conjunction if we can prove $\text{WF}_d\ f$. We have the premise that $\text{WF}_d\ (f\ x)$ from this overall induction case, which by definition of $\text{WF}_d$ can only be true if $\text{WF}_d\ f \land \text{WF}_d\ d$. Thus we can prove the left side of the conjunction, and the proof of the right side is identical with $f$ replaced by $x$ using the other inductive hypothesis.

\textbf{\textit{Case}} $e = \$i$. Our goal is

$$\text{WF}_d\ \uparrow_d \$i$$

We will handle this in 2 cases:
\begin{itemize}
\item \textit{Case} $i < d$.
$$\text{WF}_d\ \uparrow_d \$i$$
$$=\text{WF}_d\ \$i$$
This is true because \$i indices are always well-formed.
\item \textit{Case} $i \geq d$.
$$\text{WF}_d\ \uparrow_d \$i$$
$$=\text{WF}_d\ \$(i+1)$$
This is true because \$i indices are always well-formed.
\end{itemize}

\textbf{\textit{Case}} $e = \&i$. Our goal is

$$\text{WF}_d\ \uparrow_d \&i$$

We will handle our goal in 2 cases:
\begin{itemize}
\item \textit{Case} $i + 1 \neq d$.
$$\text{WF}_d\ \uparrow_d \$i$$
$$=\text{WF}_d\ \&(i+1)$$
We know from our premise that $\text{WF}_d\ \&i$ which means $i < d$, so $i + 1 \leq d$. Combining this with our knowledge from the case branch that $i + 1 \neq d$, we know that $i + 1 < d$ which means that $\&(i+1)$ is well formed.
\item \textit{Case} $i + 1 = d$.
$$\text{WF}_d\ \uparrow_d \&i$$
$$=\text{WF}_d\ \$(i+1)$$
This is true because \$i indices are always well-formed.
\end{itemize}

\textbf{\textit{Case}} $e = t$. Our goal is $\text{WF}_d\ \uparrow_d t$ which simplifies by definition of $\uparrow_d$ to $\text{WF}_d\ t$ which is true by definition of $\text{WF}_d$.

\end{proof}

\begin{lemma}\label{lemma:down-well}
$\forall e.\ \forall d.\ \text{WF}_d\ e \implies  \text{WF}_d\ \downarrow_d e$
\end{lemma}
\begin{proof}
We proceed by induction on e. The cases $e = t$, $e = (f\ x)$ and $e = \lambda. b$ are identical to those in \cref{lemma:up-well} but with $\uparrow$ replaced with $\downarrow$ so they will be omitted.\\

\textbf{\textit{Case}} $e = \$i$. Our goal is

$$\text{WF}_d\ \downarrow_d \$i$$

We will handle this in 3 cases:
\begin{itemize}
\item \textit{Case} $i < d$.
$$\text{WF}_d\ \downarrow_d \$i$$
$$=\text{WF}_d\ \$i$$
This is true because \$i indices are always well-formed.

\item \textit{Case} $i > d$.
$$\text{WF}_d\ \downarrow_d \$i$$
$$=\text{WF}_d\ \$(i-1)$$
This is true because \$i indices are always well-formed.

\item \textit{Case} $i = d$.
$$\text{WF}_d\ \downarrow_d \$i$$
$$=\text{WF}_d\ \&(i-1)$$

Since $i = d$ we know that $i - 1 < d$ which means $\&(i-1)$ is well-formed.
\end{itemize}

\textbf{\textit{Case}} $e = \&i$. Our goal is

$$\text{WF}_d\ \downarrow_d \&i$$

which simplifies to

$$\text{WF}_d\ \&(i-1)$$

We know from our premise that $\text{WF}_d\ \&i$ so $i < d$ which means $i - 1 < d$  so $\&(i-1)$ is well-formed.

\end{proof}

\ \\\ \\\noindent \textbf{The following lemma shows that upshifting is an inverse of downshifting.} In Coq this lemma is \texttt{upshift\_downshift}.

\begin{lemma}\label{lemma:up-down}
$\forall e.\ \forall d.\ \text{WF}_d\ e \implies \uparrow_d \downarrow_d e = e$
\end{lemma}
\begin{proof}

We proceed by induction on $e$.\\

\textbf{\textit{Case}} $e = \lambda.\ b$. Our goal is

$$\uparrow_d \downarrow_d \lambda.\ b = \lambda.\ b$$

By the definitions of $\uparrow_d$ and $\downarrow_d$ we can propagate them into the body of a lambda while incrementing $d$:

$$\lambda.\ \uparrow_{d+1} \downarrow_{d+1} b = \lambda.\ b$$

We also have the additional premise $\text{WF}_d\ \lambda. b$ which simplifies to $\text{WF}_{d+1}\ b$.

Since our inductive hypothesis $\forall d.\ \text{WF}_d\ b \implies \uparrow_d \downarrow_d b = b$ is universally quantified over $d$ we can instantiate it with $d+1$ in place of $d$ and discharge its precondition with our other premise $\text{WF}_{d+1}\ b$ to get $\uparrow_{d+1} \downarrow_{d+1} b = b$ and apply this to our goal to get

$$\lambda.\ b = \lambda.\ b$$

which is trivially true.\\

\textbf{\textit{Case}} $e = f\ x$. Our goal is

$$\uparrow_d \downarrow_d (f\ x) = f\ x$$

By the definitions of $\uparrow_d$ and $\downarrow_d$ we can propagate them into both sides of the application:

$$(\uparrow_d \downarrow_d f) (\uparrow_d \downarrow_d x) = f\ x$$

We also have the additional premise $\text{WF}_d\ (f\ x)$ can be broken into the two premises $\text{WF}_d\ f$ and $\text{WF}_d\ x$.

We can then apply our inductive hypotheses $\forall d.\ \text{WF}_d\ f \implies \uparrow_d \downarrow_d f = f$ and $\forall d.\ \text{WF}_d\ x \implies \uparrow_d \downarrow_d x = x$ with these premises to rewrite our goal to 

$$f\ x = f\ x$$

which is trivially true.\\

\textbf{\textit{Case}} $e = \$i$. Our goal is

$$\uparrow_d \downarrow_d \$i = \$i$$

We will handle this in 3 cases:

\begin{itemize}

\item \textit{Case} $i < d$.
$$\uparrow_d \downarrow_d \$i = \$i$$
Unfold definition of $\downarrow_d$
$$\uparrow_d \$i = \$i$$
Unfold definition of $\uparrow_d$
$$\$i = \$i$$

\item \textit{Case} $i > d$.
$$\uparrow_d \downarrow_d \$i = \$i$$
Unfold definition of $\downarrow_d$
$$\uparrow_d \$(i-1) = \$i$$
Unfold definition of $\uparrow_d$ noting that $(i-1) \geq d$
$$\$((i-1)+1) = \$i$$
Simplify
$$\$i = \$i$$

\item \textit{Case} $i = d$.
$$\uparrow_d \downarrow_d \$i = \$i$$
Unfold definition of $\downarrow_d$
$$\uparrow_d \&(i-1) = \$i$$
Unfold definition of $\uparrow_d$ noting that $(i-1)+1 = d$
$$\$((i-1)+1) = \$i$$
Simplify
$$\$i = \$i$$

\end{itemize}

\textbf{\textit{Case}} $e = \&i$. Our goal is

$$\uparrow_d \downarrow_d \&i = \&i$$

Unfolding the definition of $\downarrow_d$ yields

$$\uparrow_d \&(i-1) = \&i$$

Since we know that $\text{WF}_d\ \&i$ by our premise, we know that $i < d$ so $(i - 1) + 1 < d$ so $(i - 1) + 1 \neq d$ so $\uparrow_d$ unfolds to:

$$\&((i-1)+1) = \&i$$

which simplifies to

$$\&i = \&i$$

which is trivially true.\\

\textbf{\textit{Case}} $e = t$. Our goal is

$$\uparrow_d \downarrow_d t = t$$

Neither $\uparrow_d$ nor $\downarrow_d$ have any effect on a grammar primitive $t$, so this is trivially true.

\end{proof}

\ \\\noindent \textbf{The following lemma shows that any result of \lambdaU\ is a well-formed mapping.} In Coq this lemma is called \texttt{lu\_wf}.
\begin{lemma}\label{lemma:lu_wf}
    $\forall A.\ \forall e. \text{WFMap}\ \lambdaU(A,e)$
\end{lemma}
\begin{proof}
    This is a straightforward proof by induction. In the $\alpha_i$ and and $\hole_i$ cases \lambdaU\ adds expressions from the original expression to the mapping, which are well-formed as they can't contain $\&i$ variables. In the lambda case \textsc{DownshiftAll} shifts variables using $\downarrow$ which is guaranteed to preserve well-formedness by \cref{lemma:down-well}. In the application case \textsc{merge} simply combines mappings from the function and argument calls to \lambdaU\ which are well-formed by induction, resulting in an overall well-formed mapping.
\end{proof}

\ \\\noindent \textbf{The following lemma follows from \cref{lemma:up-down} and shows that \textsc{UpshiftAll} is an inverse of \textsc{DownshiftAll}.} In Coq this lemma is \texttt{upshift\_all\_downshift\_all}.

\begin{lemma}\label{lemma:upshift-downshift}
 $\text{WFMap}\ l \implies \textsc{UpshiftAll}(\textsc{DownshiftAll}(l)) = l$
\end{lemma}
\begin{proof}

Writing $l$ as $[\alpha_i \to\ e_i',\ \hole_j \to\  u_j',\ ...]$ and unfolding the definitions of \textsc{UpshiftAll} \textsc{DownshiftAll} we get:

$$[\alpha_i \to\ \uparrow_0 \downarrow_0 e_i',\ \hole_j \to\ \uparrow_0 \downarrow_0 u_j',\ ...] = [\alpha_i \to\ e_i',\ \hole_j \to\  u_j',\ ...] $$

This equality will hold if $\uparrow_0 \downarrow_0 u_j' = u_j'$ and $\uparrow_0 \downarrow_0 e_i' = e_i'$ for all $u_j'$ and $e_i'$. $\text{WFMap}\ l$ tells us that all expressions $u_j'$ and $e_i'$ are well-formed and thus satisfy the precondition of \cref{lemma:up-down}, which proves these equalities true.

\end{proof}

\noindent \textbf{The following two lemmas allow us to "forget" one side of a \textsc{merge} when it won't have any effect on the substitution.} In Coq the first lemma is \texttt{merge\_forget3} and the second is captured by the definition \texttt{Compatible}.
\begin{lemma}\label{lemma:merge-forget-right}
    If $l_1$ is a mapping that binds every abstraction variable and hole in $A$, then $\textsc{merge}(l_1,l_2) \subst A = l_1 \subst A$.
\end{lemma}
\begin{proof}
$l_1$ already assigns to every abstraction variable and hole, and merging in $l_2$ can't overwrite these assignments since \textsc{merge} fails when the same variable or hole is assigned to two different expressions. Thus merging in $l_2$ has no effect on the resulting expression from substitution.
\end{proof}

\begin{lemma}\label{lemma:merge-forget-left}
    If $l_2$ is a mapping that binds every abstraction variable and hole in $A$, then $\textsc{merge}(l_1,l_2) \subst A = l_2 \subst A$.
\end{lemma}
\begin{proof}
This proof is identical in structure to \cref{lemma:merge-forget-right}. $l_2$ already assigns to every abstraction variable and hole, and merging in $l_1$ can't overwrite these assignments since \textsc{merge} fails when the same variable or hole is assigned to two different expressions. Thus merging in $l_1$ has no effect on the resulting expression from substitution.
\end{proof}

\ \\\noindent \textbf{The following lemma shows that any result of \lambdaU\ satisfies the precondition of \cref{lemma:merge-forget-left} and \cref{lemma:merge-forget-right}.}

\begin{lemma}\label{lemma:binds-all}
    $\lambdaU(A,e)$ returns a mapping that binds all abstraction variables $\alpha_i$ and all holes $\hole_i$ in $A$.
\end{lemma}
\begin{proof}
    This is a straightforward proof by induction. In the $\alpha_i$ and and $\hole_i$ cases \lambdaU\ binds the variable and hole respectively, in the lambda case \textsc{DownshiftAll} has no effect on which abstraction variables are bound (and it binds all the variables in the body by induction), and in the application case \textsc{merge} results in combining the function and argument mappings (which bind all variables by induction) to result in a mapping which binds all variables that appear in both the function and argument.
\end{proof}

\subsection{Correctness Theorem}

In Coq this theorem is called \texttt{correctness}.

\begin{theorem}\label{lemma:correctness}
 $\forall A.\ \forall e.\ \lambdaU(A,e) \subst  A = e$
\end{theorem}
\begin{proof}

We proceed by induction over the proof tree of $\lambdaU(A,e)$.\\

\textit{Case} \textsc{U-AbsVar}. Our goal is $\lambdaU(A,\alpha) \subst  \alpha = e$ which simplifies by definition of $\lambdaU$ to $[\alpha \to e] \subst  \alpha = e$ which is true by the definition of substitution.\\

\textit{Case} \textsc{U-Hole}. Our goal is  $\lambdaU(A,\hole) \subst  \hole = e$ which simplifies by definition of $\lambdaU$ to $[\hole \to e] \subst  \hole = e$ which is true by the definition of substitution.\\

\textit{Case} \textsc{U-App}. Our goal is

$$\lambdaU(A_1\ A_2,\ e_1\ e_2) \subst (A_1\ A_2) = e_1\ e_2$$

This simplifies by definition of $\lambdaU$ to 

$$ \textsc{merge}(\lambdaU(A_1,e_1),\lambdaU(A_2,e_2)) \subst (A_1\ A_2) = e_1\ e_2$$

We can distribute the substitution over the application (by definition of substitution in \cref{fig:subst-upshift} (Left)) to get

\begin{align*}
(&(\textsc{merge}(\lambdaU(A_1,e_1),\lambdaU(A_2,e_2)) \subst A_1)\\ &(\textsc{merge}(\lambdaU(A_1,e_1),\lambdaU(A_2,e_2)) \subst A_2))\\
&= e_1\ e_2
\end{align*}

We can then forget the irrelevant part of each merge using \cref{lemma:merge-forget-left} and \cref{lemma:merge-forget-right} where the precondition (binding all abstraction variables and holes) is satisfied by \cref{lemma:binds-all}:

\begin{align*}
(&(\lambdaU(A_1,e_1) \subst A_1)\\
&(\lambdaU(A_2,e_2) \subst A_2))\\
&= e_1\ e_2
\end{align*}

We can then finish the proof by directly applying our inductive hypotheses $\lambdaU(A_1,e_1) \subst A_1 = e_1$ and $\lambdaU(A_2,e_2) \subst A_2 = e_2$.\\

\textit{Case} \textsc{U-Lam}. Our goal is

$$\lambdaU(\lambda.\ A,\lambda.\ e) \subst \lambda.\ A = \lambda.\ e$$

This simplifies by definition of $\lambdaU$ to 

$$\textsc{DownshiftAll}(\lambdaU(A,e)) \subst \lambda.\ A = \lambda.\ e$$

We can move the substitution under the lambda while inserting $\textsc{UpshiftAll}$ as per the definition of substitution

$$\lambda.\ \textsc{UpshiftAll}(\textsc{DownshiftAll}(\lambdaU(A,e))) \subst A = \lambda.\ e$$

Since we know that any result of \lambdaU\ is well-formed by \cref{lemma:lu_wf} we can apply \cref{lemma:upshift-downshift} to get

$$\lambda.\ \lambdaU(A,e) \subst A = \lambda.\ e$$

We can then apply our inductive hypothesis $\lambdaU(A,e) \subst A = e$ to get

$$\lambda.\ e = \lambda.\ e$$

and we are done.\\

\textit{Case} \textsc{U-Same}. Our goal is $\lambdaU(e,e) \subst  e = e$ which simplifies by definition of $\lambdaU$ to $[] \subst  e = e$ which is trivially true as the empty substitution has no effect.

\end{proof}

\end{document}